\documentclass[twocolumn]{svjour3}         
\smartqed  

\usepackage{graphicx}
\usepackage{amssymb,amsmath}
\usepackage{hyperref}
\usepackage{breakurl}
\usepackage{graphicx}
\usepackage{xcolor}
\usepackage{enumitem}

\usepackage[nolineno]{lgrind}

\usepackage{subcaption}
\usepackage{caption}
\usepackage[linesnumbered,noline,ruled,noend,algosection]{algorithm2e}
\usepackage{chngcntr}
\counterwithin{figure}{section}
\counterwithin{table}{section}

\def\ojoin{\setbox0=\hbox{$\Join$}%
  \rule[0.05ex]{.33em}{0.4pt}\llap{\rule[1ex]{.33em}{0.4pt}}}
\def\leftouterjoin{\mathbin{\ojoin\mkern-7.1mu\Join}}
\newtheorem{prop}{Property}
\counterwithin{definition}{section}
\counterwithin{prop}{section}
\counterwithin{lemma}{section}

\usepackage[final,layout={index,margin}]{fixme}

\begin{document}

\title{Algorithms and Analysis for the SPARQL Constructs}


\author{Medha Atre}


\institute{Medha Atre \at
              Dept. of Computer Science and Engineering,\\
              Indian Institute of Technology, Kanpur, India\\
              \email{medha.atre@gmail.com} 
}

\date{Received: date / Accepted: date}

\maketitle

\begin{abstract}
As Resource Description Framework (RDF) is becoming a popular data modelling 
standard, the challenges of efficient processing of \textit{Basic Graph Pattern} 
(BGP) SPARQL queries (a.k.a. SQL inner-joins) have been a focus of the research 
community over the past several years. 
In our recently published work we brought community's attention to another 
equally important component of SPARQL, i.e., OPTIONAL pattern queries (a.k.a. 
SQL left-outer-joins). We proposed novel optimization techniques 
-- first of a kind -- 
and showed experimentally that our techniques perform better for the 
\textit{low-selectivity} queries, and give at par performance for the highly 
selective queries, compared to the state-of-the-art methods.

BGPs and OPTIONALs (BGP-OPT) make the basic building blocks of the SPARQL query 
language. Thus, in this paper, treating our BGP-OPT query optimization 
techniques as the \textit{primitives}, we extend them to handle other broader 
components of SPARQL such as such as UNION, FILTER, and DISTINCT. We 
mainly focus on the \textit{procedural} (algorithmic) aspects of these 
extensions. We also make several important observations about the 
structural aspects of complex SPARQL queries with any intermix of these clauses, 
and \textit{relax} some of the constraints regarding the \textit{cyclic} 
properties of the queries proposed earlier. We do so without affecting the 
correctness of the results, thus providing more flexibility in using the 
BGP-OPT optimization techniques.
\end{abstract}

\section{Introduction} \label{sec:intro}
Resource Description Framework (RDF) \cite{rdf} is being used as the standard 
for representing \textit{semantically linked} data on the web as well as for 
other domains such as biological networks, e.g. UniProt RDF network by Swiss 
Institute of 
Bioinformatics\footnote{\url{http://www.uniprot.org/format/uniprot_rdf}}.
RDF is a directed edge-labeled multi-graph, where each unique edge (S P 
O) is called a \textit{triple} -- P is the label on the edge from the node S to 
node O, and SPARQL \cite{sparql} is the standard query language for it.

SPARQL provides various syntactic constructs to form \textit{structured} queries 
over RDF graphs. These constructs have a close similarity to their SQL 
counterparts. For instance, Basic Graph Patterns (BGP) of SPARQL (or 
\textit{TriplesBlock} as referred to in the SPARQL grammar) are similar to 
the SQL inner-joins ($\Join$). OPTIONAL patterns of SPARQL 
(\textit{OPTIONALGraphPattern} in the SPARQL grammar) are similar to 
the left-outer-joins ($\leftouterjoin$). FILTERs of SPARQL makes up for the 
SQL LIKE clause and various other selection conditions. UNIONs ($\cup$) and 
DISTINCTs of SPARQL are similar to their SQL counterparts.
\textit{GroupGraphPattern} of the SPARQL grammar consists of BGP, 
OPTIONAL, UNION, and FILTER components, and like SQL, SPA-RQL grammar too 
allows nested queries with a complex intermix of these query constructs. Since 
there is an equivalence between SPARQL and SQL constructs, we will use these 
terms interchangeably in the rest of the text.

BGP queries make the building blocks of SPARQL, and just like SQL 
inner-joins, they are \textit{associative} and \textit{commutative}, i.e., a 
change in the order of joins among the triple patterns does not change the 
final results -- thus allowing a \textit{reorderability} among the BGP triple 
patterns. Owing to these similarities, the RDF and SPA-RQL community has 
adopted methods of SQL inner-join optimization, and have taken them further with 
the novel ideas of RDF graph indexing \cite{bitmatwww10,triad,rdf3x,triplebit}. 
However, optimization of SPARQL queries with other query constructs poses 
additional challenges, because they restrict the 
\textit{reorderability} of the triple patterns across various BGPs. Given below 
is an OPTIONAL pattern query from \cite{atresigmod15}.

\vspace{2mm}
\begin{lgrind}
\Head{}
{\small
\L{\LB{\K{Q1}:_\K{SELECT}_?\V{friend}_?\V{sitcom}}}
\L{\LB{}\Tab{6}{\K{WHERE}_\{}}
\L{\LB{}\Tab{8}{:\V{Jerry}_:\V{hasFriend}_?\V{friend}_.}}
\L{\LB{}\Tab{8}{\K{OPTIONAL}_\{}}
\L{\LB{}\Tab{10}{?\V{friend}_:\V{actedIn}_?\V{sitcom}_.}}
\L{\LB{}\Tab{10}{?\V{sitcom}_:\V{location}_:\V{NYC}_.}}
\L{\LB{}\Tab{8}{\}}}
\L{\LB{}\Tab{6}{\}}}
}
\end{lgrind}
\vspace{2mm}

This query asks for all friends of \textit{:Jerry} that have acted in a sitcom 
located in \textit{:NYC}. In this query, let (\textit{:Jerry 
:hasFriend ?friend}) be $T_1$, (\textit{?friend :actedIn ?sitcom}) $T_2$, and 
(\textit{?sitcom :location :NYC}) $T_3$. $T_1$ makes a left-outer-join over 
\textit{?friend} with $T_2$. $T_2$ and $T_3$ make an inner-join between 
them over \textit{?sitcom}. The query can be expressed as $(Query = T_1 
\leftouterjoin (T_2 \Join T_3))$. If we consider triple patterns to be 
equivalent to relational tables, then $T_1$ forms a BGP (say $P_1$) with just 
one triple pattern, and $(T_3 \Join T_4)$ forms another BGP (say $P_2$). Note 
that we emphasized the order of joins by putting the join $(T_2 \Join T_3)$ in a 
bracket to indicate that this inner-join must be evaluated before the 
left-outer-join between $T_1$ and $T_2$ for the correct results. This is 
because $T_1 \leftouterjoin (T_2 \Join T_3) \neq (T_1 \leftouterjoin T_2) \Join 
T_3 \neq (T_1 \leftouterjoin T_2) \leftouterjoin T_3$ (we will show this with a 
toy example in Section \ref{sec:nullbm}). Inner and left-outer joins are 
non-reorderable, so when we have a query with an intermix of other query 
operators too such as UNIONs, FILTERs in addition to the OPTIONALs, this poses 
additional restrictions on reorderability. Consider the following query with 
an intermix of BGPs, OPTIONALs, and UNIONs\footnote{Unlike SQL, SPARQL 
standards allows UNIONs between results of different \textit{arity}.}.

\vspace{2mm}
\begin{lgrind}
\Head{}
{\small
\L{\LB{\K{Q2}:_\K{SELECT}_?\V{friend}_?\V{sitcom}}}
\L{\LB{}\Tab{6}{\K{WHERE}_\{}}
\L{\LB{}\Tab{8}{:\V{Jerry}_:\V{hasFriend}_?\V{friend}_.}}
\L{\LB{}\Tab{8}{\{}}
\L{\LB{}\Tab{12}{\{?\V{friend}_:\V{actedIn}_?\V{sitcom}_.}\}}
\L{\LB{}\Tab{12}{\K{UNION}}}
\L{\LB{}\Tab{12}{\{?\V{friend}_:\V{hasFriend}_?\V{friend2}_.}}
\L{\LB{}\Tab{12}{?\V{friend2}_:\V{actedIn}_?\V{sitcom}_.}\}}
\L{\LB{}\Tab{8}{\}}}
\L{\LB{}\Tab{8}{\K{OPTIONAL}_\{}}
\L{\LB{}\Tab{10}{?\V{sitcom}_:\V{hasDirector}_?\V{dir}_.}}
\L{\LB{}\Tab{10}{?\V{sitcom}_:\V{location}_:\V{NYC}_.}}
\L{\LB{}\Tab{8}{\}\}}}
}
\end{lgrind}
\vspace{2mm}

This query asks for all the friends and friends-of-friends of \textit{:Jerry} 
who have acted in \textit{any} sitcom, and \textit{optionally} it asks for the 
directors of the respective sitcoms if their location was \textit{NYC}. We have 
in all six triple patterns in this query. Numbering them $T_{1..6}$ from top 
to bottom, the query can be expressed as $Q = (T_1 \Join (T_2 \cup (T_3 \Join 
T_4))) \leftouterjoin (T_5 \Join T_6)$. These triple patterns form four BGPs in 
the query, which are as follows: $P_1 = T_1, P_2 = T_2, P_3 = (T_3 \Join T_4), 
P_4 = (T_5 \Join T_6)$, and then the query can be expressed as $Q = (P_1 
\Join (P_2 \cup P_3)) \leftouterjoin (P_4)$. Note that we cannot do the 
left-outer join between $T_2$ and $T_5, T_6$ before evaluating $P_4 =$ ($T_5 
\Join T_6$) and the UNION $P_2 \cup P_3$.

Analysis of the real world SPARQL queries shows that queries with an intermix 
of BGP, OPTIONALs, UNIONs, FILTERs indeed constitute over 94\% the query logs
\cite{usewod11,swim,manvmachine,practsparql}, and thus make these other 
constructs like OPTIONAL, UNION, FILTER non-negligible from the query processing 
and performance optimization perspective. In our previous work 
\cite{atresigmod15} we focused on the OPTIONAL pattern queries (referred to 
henceforth as OPT queries), and proposed novel techniques for the 
optimization of these queries. Our techniques extended the ideas of 
\textit{nullification} and \textit{best-match} (or 
\textit{Generalized-Outerjoin}) operators as proposed in 
\cite{rao2,rao1,galindo-legaria2}, and sho-wed that for \textit{acyclic} 
queries we can reduce the candidate triples to \textit{minimal} using the 
\textit{semi-join} based \textit{pruning} \cite{semij1,semij2,ullman}, and avoid 
the nullification and best-match operations altogether (see the lemmas 
in \cite{atresigmod15}). Since BGP-OPT make the building blocks of SPARQL 
queries, in this paper we mainly show that our BGP-OPT optimization techniques 
can be used as \textit{primitives} to evaluate queries with an intermix of the 
various SPARQL query constructs. Much of the research based systems developed 
for the optimization of SPARQL queries have only handled the BGP component 
\cite{bitmatwww10,rdf3x,triplebit}, and they either do not handle other SPARQL 
constructs, or rely on na\"{\i}ve ways of processing them. The SPARQL processing 
systems based on relational databases, such as MonetDB or Virtuoso, just 
translate the SPARQL queries into their SQL counterparts, by assuming that the 
RDF graph is stored in the relational tables.

In the light of this, we propose to use our BGP-OPT evaluation techniques as 
\textit{primitive building blocks} to cover a broader spectrum of the SPARQL 
queries. While doing so, we focus on the \textit{procedural} (algorithmic) 
aspects of using BGP-OPT techniques than the empirical aspects, because our 
previous work has already established the usefulness of our BGP-OPT evaluation 
techniques -- especially for the \textit{low-selectivity}  
que-ries\footnote{Queries which need to process a large amount of data have low 
selectivity and vice versa.}.
In this paper, we make the following main contributions.
\begin{enumerate}
 \item We propose a new method of forming the \textit{Graph of Supernodes} 
(GoSN) that enhances our previously proposed method \cite{atresigmod15} 
(Section \ref{sec:gosn}).
 \item Using the above mentioned new way of constructing the GoSN, we show 
that the condition of \textit{acyclicity} of \textit{Graph of Tables} (GoT) of a 
BGP-OPT query can be relaxed in some cases in addition to the conditions given 
in \cite{atresigmod15} (Sections \ref{sec:nullmin} and \ref{sec:cyclic}).
  \item We propose a way of methodically using the BGP-OPT query 
optimization techniques for the queries with an intermix of UNION and FILTER 
clauses \textit{without} evaluating each query in the 
\textit{UNION normal form} (UNF) individually as proposed in 
\cite{atresigmod15} (Sections \ref{sec:union} and \ref{sec:filter}).
\item We also discuss handling of the DISTINCT clause  with any intermix of 
these query clauses (Section \ref{sec:distinct}).
  \item In the context of UNION and FILTER, we bring to the light implications 
of ``NULLs'', and their semantics for the \textit{nullification} and 
\textit{best-match} operators.
\item Since there is a close match between SPARQL query operators and SQL, our 
techniques and insights can be useful for their SQL counterparts too, with 
appropriate indexing and data representation methods in the relational setting.
\end{enumerate}

\section{Graph of Supernodes} \label{sec:gosn}
In our previous work \cite{atresigmod15}, we had outlined a way of capturing 
a SPARQL query with an intermix of BGP and OPTIONAL patterns using the 
\textit{Graph of Supernodes} (GoSN). For the sake of completeness of the text, 
here we first describe the process of GoSN construction, and then elaborate on 
the \textit{new} enhancements. These enhancements help in our 
propositions regarding the relaxation of the \textit{nullification} and 
\textit{best-match} operations, based on the \textit{cyclic} 
properties of a query. For this construction of GoSN, we focus only on the BGP 
OPT patterns without any other SPARQL constructs. They serve as the 
\textit{primitives} for applying the BGP-OPT query processing techniques for a 
broader range of queries with any intermix of UNIONs, FILTERs, and DISTINCT as 
elaborated in Section \ref{sec:discuss}.

A BGP-OPT pattern query connects multiple BGP patterns with each 
other using one or more \textit{OPTIONAL} keywords. Connecting BGP patterns 
($\Join$) with OPTION-AL clauses ($\leftouterjoin$) introduces restrictions on 
the order of joining the triple patterns across various BGPs (refer to the 
example of reorderability given in Section \ref{sec:intro}).
A GoSN is constructed from a BGP-OPT query using \textit{supernodes}, and 
\textit{unidirectional} or \textit{bidirectional} edges, as follows.

\begin{figure}[h]
       \centering
      \includegraphics[scale=0.4]{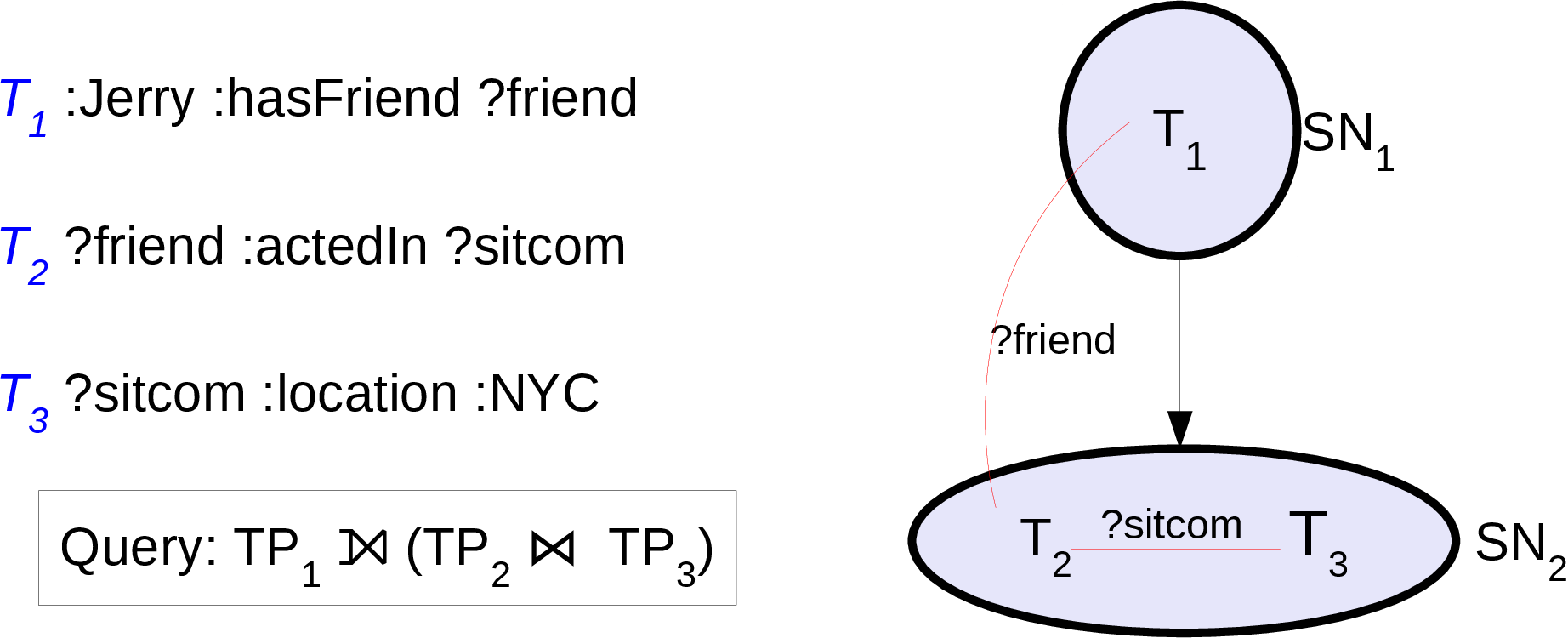}
      \captionsetup{singlelinecheck=off}
  \caption{GoSN for Q1 in Section \ref{sec:intro}}\label{fig:qgraph}
\end{figure}

\begin{figure}[h]
\begin{center}
               \includegraphics[scale=0.4]{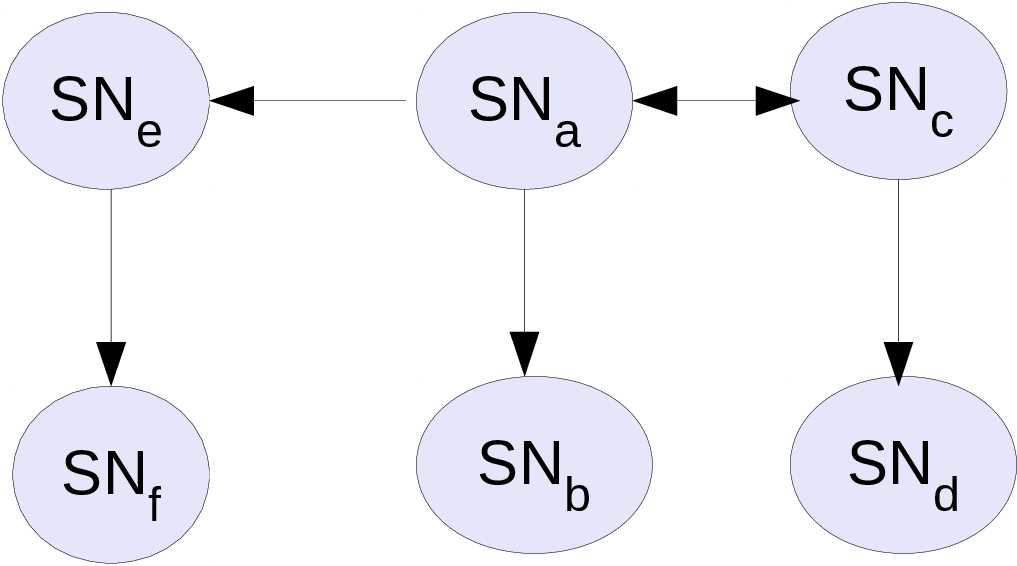}
               \captionsetup{singlelinecheck=off}
    \caption[.]{GoSN for $((P_a \leftouterjoin P_b) \Join (P_c \leftouterjoin P_d)) \leftouterjoin (P_e 
\leftouterjoin P_f)$} \label{fig:qgrex}
\end{center}
\end{figure}

\textbf{Supernodes:}
In a BGP-OPT query of the form $(P_1 \leftouterjoin P_2)$, $P_1$ and $P_2$ are 
patterns that may in turn have nested OPTIONAL patterns inside them, e.g., $P_1 
= (P_3 \leftouterjoin P_4)$, or either of $P_1$ and $P_2$ can be OPT-free. 
Generalizing it, if a pattern $P_i$ does not have any OPTIONAL pattern nested 
inside it, it is an \textit{OPT-free} BGP or simply a BGP. From the given 
nested BGP-OPT query, first we extract all such BGPs, and construct a 
\textit{supernode}  ($SN_i$) for each $P_i$. The triple patterns in $P_i$ are 
\textit{encapsulated} in $SN_i$.

Since BGPs are equivalent to SQL inner-joins, and OPTIONAL patterns are 
equivalent to SQL left-outer-joins, we serialize a nested BGP-OPT query 
considering its BGPs and $\Join$ (inner-join), $\leftouterjoin$ 
(left-outer-join) operators using proper parentheses. \fxnote{Give an example 
consistent with intro.} E.g., we serialize \textbf{Q1} in Section 
\ref{sec:intro} as ($P_1 \leftouterjoin P_2$), where $P_1$ and $P_2$ are 
OPT-free BGPs, $SN_1$ of $P_1$ encapsulates just $T_1$, and $SN_2$ of $P_2$ 
encapsulates $T_2$ and $T_3$ (see Figure \ref{fig:qgraph}). 

\textbf{Unidirectional edges:}
From the serialized query, we consider each OPT pattern of the type $P_m 
\leftouterjoin P_n$. $P_m$ or $P_n$ may have nested OPT-free BGPs inside them.
Using the serialized-parenthesized form, we identify the \textit{leftmost} 
OPT-free BGPs nested inside $P_m$ and $P_n$ each. E.g., consider the example 
query given in Fig. \ref{fig:qgrex}. If $P_m = ((P_a 
\leftouterjoin P_b) \Join (P_c \leftouterjoin P_d))$, and $P_n = (P_e 
\leftouterjoin P_f)$, $P_a$ and $P_e$ are the \textit{leftmost} OPT-free BGPs 
in $P_m$ and $P_n$ respectively, and $SN_a$ and $SN_e$ are their supernodes. We 
add a directed edge $SN_a \rightarrow SN_e$. If either $P_m$ or $P_n$ does not 
nest any OPT-free BGPs inside it, we treat the very pattern as the 
\textit{leftmost} for adding a directed edge. With this procedure, we can treat 
OPTIONAL patterns in a query in any order, but for all practical purposes, we 
start from the \textit{innermost} OPTIONAL patterns, and recursively go on 
considering the outer OPTIONAL patterns using the parentheses in the serialized 
query. E.g., if a serialized query is ($(P_a \leftouterjoin P_b)$ $\Join$
$(P_c \leftouterjoin P_d)$) $\leftouterjoin$ ($P_e \leftouterjoin P_f$), with 
$P_{a}...P_{f}$ as OPT-free BGPs, we add directed edges as follows: (1) $SN_a 
\rightarrow SN_b$, (2) $SN_c \rightarrow SN_d$, (3) $SN_e \rightarrow SN_f$, 
(4) $SN_a \rightarrow SN_e$.

\textbf{Bidirectional edges:}
Next we consider each inner-join of type $P_x \Join P_y$ in a serialized query.
If $P_x$ or $P_y$ has nested OPTIONALs inside, we add a bidirectional edge 
between the supernodes of \textit{leftmost} OPT-free BGPs.
E.g., if $P_x = (P_a \leftouterjoin P_b)$, and $P_y =(P_c \leftouterjoin P_d)$,
we add a bidirectional edge $SN_a \leftrightarrow SN_c$. If $P_x$ or $P_y$ does 
not nest any OPTIONALs inside it, we consider the very pattern to be the 
\textit{leftmost} for adding a bidirectional edge. We add bidirectional edges 
starting from the \textit{innermost} inner-joins ($\Join$) using the 
parentheses 
in the serialized query, and recursively go on considering the outer ones, 
until 
no more bidirectional edges can be added. Considering the same example given 
under unidirectional edges, we add a bidirectional edge between $SN_a 
\leftrightarrow SN_c$. The \textit{graph of supernodes} (GoSN) for this example 
is shown in Figure \ref{fig:qgrex}. Thus we completely capture the nesting of 
BGPs and OPTIONALs in a query using this GoSN.

\subsection{Nomenclature} \label{sec:nomen}
Next, we introduce nomenclatures with respect to the supernodes in a GoSN 
and the OPTIONAL patterns in a SPARQL query.

\textbf{Master-Slave:} In an OPTIONAL pattern $P_c \leftouterjoin P_d$,
we call pattern $P_c$ to be a \textit{master} of $P_d$, and $P_d$ a 
\textit{slave} of $P_c$. This master-slave relationship is \textit{transitive},
i.e., if a supernode $SN_f$ is \textit{reachable} from another supernode $SN_c$ 
by following \textit{at least} one unidirectional edge in GoSN  
($SN_c...\rightarrow...SN_f$), then $SN_c$ is  called a master of $SN_f$
(see Figure \ref{fig:qgrex}).

\textbf{Peers:} We call two supernodes to be peers if they are connected to each 
other through a bidirectional edge, or they can be \textit{reached} from each 
other by following \textit{only} bidirectional edges in GoSN, e.g., $SN_a$ and 
$SN_c$ in Figure \ref{fig:qgrex}.

\textbf{Absolute masters:} Supernodes that are not reachable from any other
supernode through a path involving any unidirectional edge are called the
\textit{absolute masters}, e.g., $SN_a$ and $SN_c$ in Figure \ref{fig:qgrex}
are absolute masters.

These master-slave, peer, and absolute master nomenclatures and relationships 
apply to any triple patterns enclosed within the respective supernodes too.

\textbf{Well-designed patterns:} As per the definition given by P\'{e}rez et al, 
a \textit{well-designed} OPT query is -- for every subpattern of type $P' = P_k 
\leftouterjoin P_l$ in the query, if a join variable $?j$ in $P_l$ appears 
outside $P'$, then $?j$ also appears in $P_k$. A query that violates this 
condition is said to be \textit{non-well-designed}.

For the scope of the text in this paper, we mainly consider 
\textit{well-designed} queries, because they occur most commonly for RDF graphs, 
and remain unaffected by the differences between SPARQL and SQL semantics over 
the treatment of \textit{NULLs}. Our previous text \cite{atresigmod15} discusses 
non-well-designed queries and their effect on the treatment of NULLs. We request 
the interested readers to refer to those (please see Appendices B and C in 
\cite{atresigmod15}).

\subsection{Graph of Triple Patterns (GoT)} \label{sec:transform}
During the GoSN construction, we only added connections between the supernodes 
formed out of the BGPs in a query based on the structural semantics of the 
given query.
Next we add \textit{labeled undirected} edges between triple patterns as 
follows. If two triple patterns share one or more join variables among them,
and are in direct master-slave hierarchy, or are part of the same supernode, we 
add an undirected edge between them. For instance, let $T_i$ and $T_j$ share 
a join variable $?j$. If $T_i \in SN_a, T_j \in SN_b, SN_a \rightarrow SN_b$, 
or if $T_i \in SN_a, T_j \in SN_a$, then we add an undirected edge between 
$T_i$ and $T_j$, with the edge label $?j$. Recall that the triple patterns 
encapsulated in the supernodes share the same master-slave or peer hierarchy as 
their respective supernodes.
These undirected edges among the triple patterns create a graph of 
triple patterns (GoT) \cite{atresigmod15}. The GoT for Q1 in Section 
\ref{sec:intro} is shown by ``red'' connecting edges in Figure 
\ref{fig:qgraph}. The edge labels in the GoT are not shown to avoid cluttering 
the figure.

\begin{definition}
If the graph of tables (GoT) of a BGP-OPT query is connected, then the query is 
free from any Cartesian joins, and is considered to be a \textbf{connected 
query}.
\end{definition}

E.g., following is an example of a \textit{non-connected} query (Cartesian 
join), because the triple pattern \textit{(?actor :livesIn :LA)} does not have 
any shared variable with the other two triple patterns \textit{(:Jerry 
:hasFriend ?friend)} and \textit{(?friend :name ?name)}.

\vspace{2mm}
\begin{lgrind}
{\small
\L{\LB{\K{SELECT}_?\V{friend}_?\V{name}_?\V{actor}}}
\L{\LB{\K{WHERE}_\{}}
\L{\LB{}\Tab{5}{:\V{Jerry}_:\V{hasFriend}_?\V{friend}_.}}
\L{\LB{}\Tab{5}{?\V{friend}_:\V{name}_?\V{name}_.}}
\L{\LB{}\Tab{5}{\K{OPTIONAL}_\{}}
\L{\LB{}\Tab{7}{?\V{actor}_:\V{livesIn}_:\V{LA}_.}}
\L{\LB{}\Tab{5}{\}\}}}
}
\end{lgrind}
\vspace{2mm}

Let us consider a subgraph of this GoT consisting of only the triple patterns 
encapsulated inside all the absolute master supernodes and all the undirected 
edges incident on them.

\begin{prop} \label{prop:got}
If a query is well-designed \textbf{and} connected, then the subgraph of GoT 
consisting only of triple patterns within the absolute masters is always 
connected. 
\end{prop}
\begin{prop}\label{prop:uniedges}
In a well-designed connected query, a slave supernode never has more than one 
incoming unidirectional edge.
\end{prop}

\begin{figure}
\centering
 \includegraphics[scale=0.4]{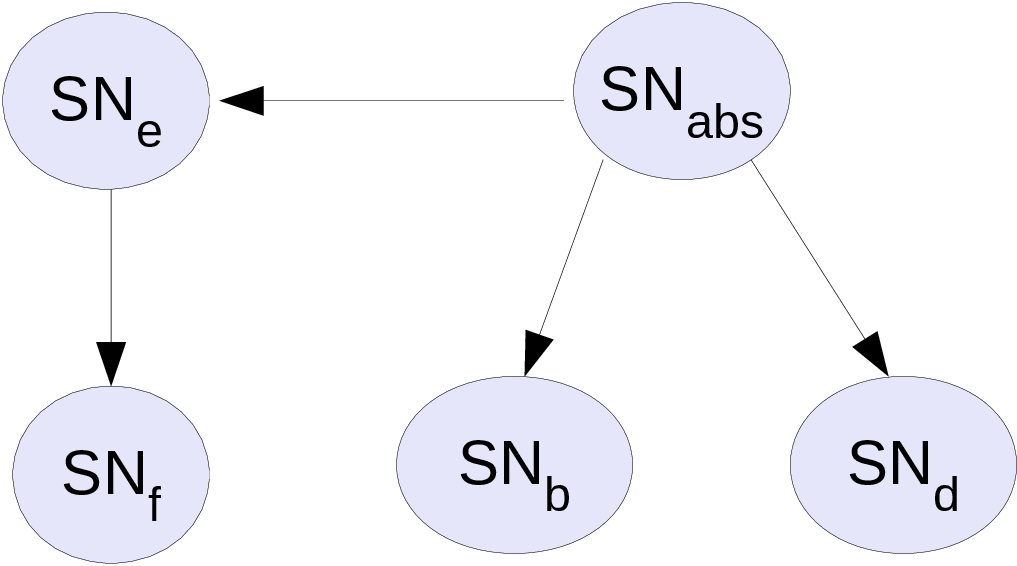}
 \caption{GoSN of Figure \ref{fig:qgrex} after coalescing absolute 
masters}\label{fig:aftcoalesce}
\end{figure}

Observing Properties \ref{prop:got} and \ref{prop:uniedges}, we coalesce all 
the absolute masters of GoSN to form a single absolute master supernode of the 
GoSN -- $SN_{abs}$. While doing so, we remove any bidirectional edges incident 
on the coalesced absolute master supernodes. For every unidirectional edge 
between a coalesced absolute master and its slave, a unidirectional edge is 
added between $SN_{abs}$ and the corresponding slave. Figure 
\ref{fig:aftcoalesce} shows the transformed GoSN of the original GoSN shown in 
Figure \ref{fig:qgrex}, after coalescing absolute masters $SN_a$ and $SN_c$. 
\fxnote{Draw figure here.}

\section{Acyclicity and Minimality} \label{sec:prelim}
In this section, first we define the \textit{acyclicity} of a SPARQL 
(equivalently SQL) query, and then discuss the \textit{minimality} of triples 
(tuples), and the effect of the cyclic properties of a query on it.

\subsection{Acyclicity of Queries} \label{sec:acyclic}
For the concept of acyclicity of a query we take into consideration the 
graph of tables (GoT). We define the equivalence classes of edges of GoT as 
follows.

\begin{definition} \label{def:eqvclass}
For every triple pattern, its incident edges in GoT are put in 
\textbf{equivalence classes} such that all the edges in a given equivalence 
class have either same edge labels or their edge labels are subset of some 
other edge's label in the same class. E.g., if $T_i$ has three edges incident 
on it with labels $\{?a\}$, $\{?a, ?b\}$, $\{?c\}$, then $\{\{?a\}, \{?a, 
?b\}\}$ make one equivalence class and $\{\{c\}\}$ in another. 
\end{definition}

A triple pattern is called a \textit{leaf} if it has only one equivalence class 
among its incident edges. An acyclic query is then defined as follows. If we 
recursively remove leaf triple patterns, and edges incident on them from a GoT, 
and then if we are left with an empty GoT at the end, then the query is 
acyclic. The set of leaf triple patterns are chosen recursively in each round 
after previous leaves and their incident edges are removed. This process is 
reminiscent of \textit{GYO-reduction} \cite{ullman}. GYO-reduction assumes a 
\textit{hypergraph} where each attribute in a table is a node, and a 
\textit{hyperedge} represents a table. However, to be consistent with our 
representation of GoT and GoSN, we have formulated this definition of 
acyclicity instead of using GYO-reduction.

\subsection{Minimality of triples}
The triples associated with a triple pattern (or tuples in a table) are said to 
be \textit{minimal} for the given join (BGP or BGP-OPT) query, if every triple 
(tuple) is part of one or more final join results of the query. There does not 
exist any triple that gets eliminated as a result of its join with another 
triple (associated with another triple pattern).
Consider the same query given in Figure \ref{fig:qgraph}, along with the sample
data associated with it in Figure \ref{fig:nullbest}. $T_2$ \textit{?friend 
:actedIn ?sitcom} has five triples associated with it -- (1) \textit{:Larry 
:actedIn :CurbYourEnthu}, (2) \textit{:Julia :actedIn :Seinfeld}, (3) 
\textit{:Julia :actedIn :Veep}, (4) \textit{:Julia :actedIn 
:NewAdvOldChristine}, (5) \textit{:Julia :actedIn :CurbYourEnthu}. But they are 
\textit{not minimal} for this query, because after $T_2$'s join with $T_3$ 
(\textit{?sitcom :location :NYC}), all tuples but \textit{:Julia 
:actedIn :Seinfeld} associated with $T_2$ get eliminated.

\section{Nullification and Best-match} \label{sec:nullbm}
\fxnote{Draw a block diagram of queries covered by BitMat algo including UNIONs 
etc.}

\begin{figure*}[t]
    \centering
        \includegraphics[width=6in,height=3.2in]{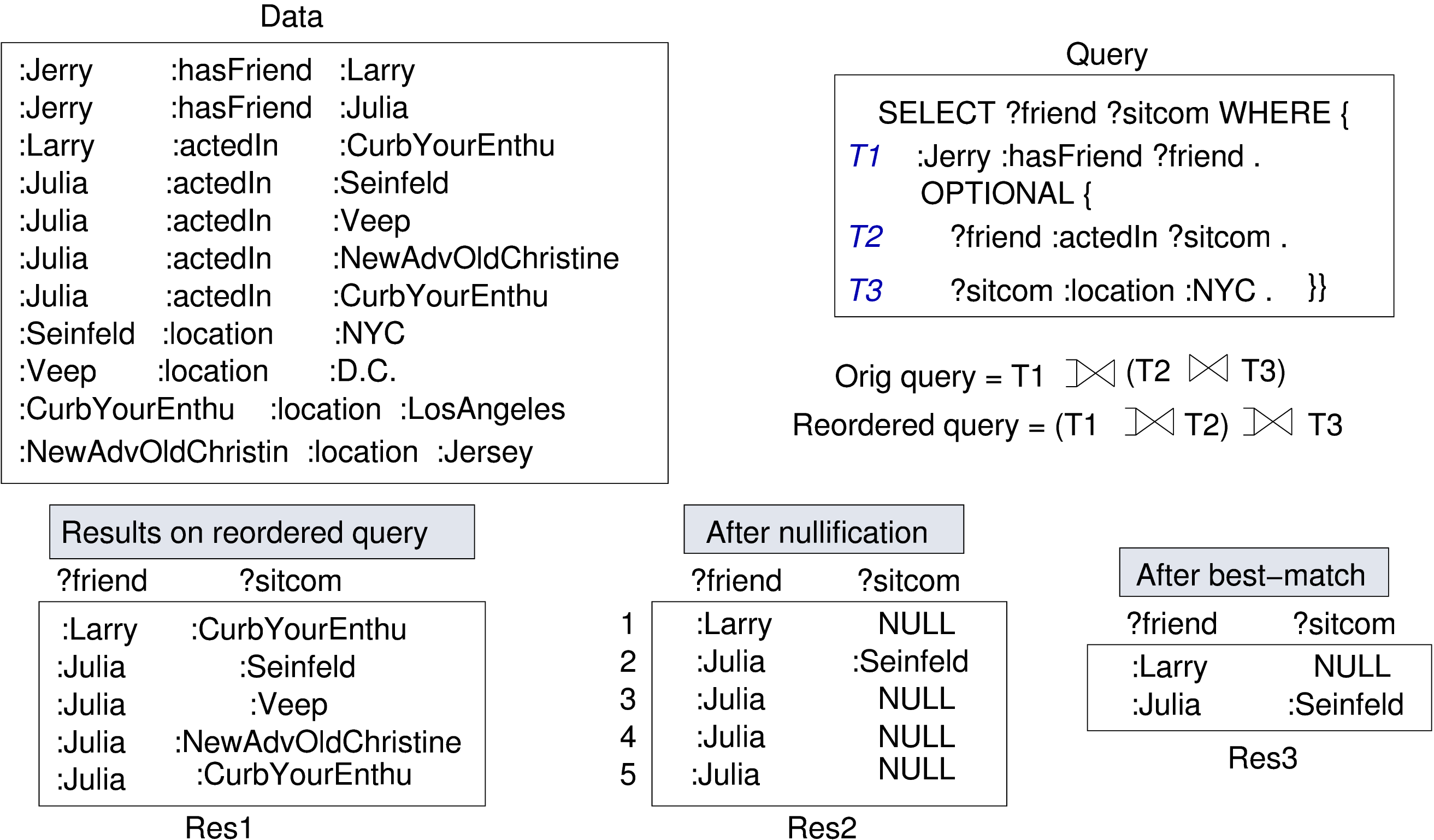}
       \caption{Nullification and best-match example} \label{fig:nullbest}
\end{figure*}

SPARQL BGP queries are analogous to the SQL inner-join queries, and hence the 
joins over the triple patterns in BGP queries are \textit{associative} and 
\textit{commutative}, i.e., a change in the order of joins between the triple 
patterns does not change the query results. SPARQL queries with OPTIONAL 
patterns are analogous to the SQL left-outer-joins, and hence they are 
\textit{not} 
associative or commutative. An example of such non-reorderable query is given 
in Section \ref{sec:intro}. However, reorderability of joins is a powerful 
feature that enables the query optimizer to explore many query plans. Hence, Rao 
et al and Galindo-Legaria, Rosenthal proposed ways of \textit{reordering} 
intermixed inner and left-outer joins by using additional operators 
\textit{nullification} and \textit{best-match} (or \textit{Generalized 
Outerjoin}) \cite{galindo-legaria1,galindo-legaria2,rao2,rao1}.

For the completeness of the text, first we will briefly see how nullification 
and best-match operators work with the same example as given in 
\cite{atresigmod15}. For more details of these operators, we request the 
interested readers to refer to 
\cite{rao1,rao2,galindo-legaria1,galindo-legaria2}.
Consider the same query given in Figure \ref{fig:qgraph}, along with the sample
data associated with it in Figure \ref{fig:nullbest}.
\textit{:NYC} has been the location for a lot of American sitcoms,
and a lot of actors have acted in them (they are not shown in the sample data 
for conciseness). But, among all such actors, \textit{:Jerry} has only two 
friends, \textit{:Julia} and \textit{:Larry}. Hence, $T_1$ is more 
\textit{selective} than $T_2$ and $T_3$. A left-outer-join reordering 
algorithm as proposed in \cite{rao1,galindo-legaria2} will typically reorder 
these joins as $(T_1 \leftouterjoin T_2) \leftouterjoin T_3$.
Due to this reordering, all four sitcoms that \textit{:Julia} has acted in show 
up as the bindings of \textit{?sitcom} (see Res1 in Fig. \ref{fig:nullbest}),
although only \textit{:Seinfeld} was located in the \textit{:NYC}.
To fix this, \textit{nullification} operator is used, which ensures that 
variable bindings across the reordered joins are consistent with the original 
join order in the query (see Res2 in Figure \ref{fig:nullbest}).

The \textit{nullification} operation caused results that are \textit{subsumed} 
within other results. A result $r_1$ is said to be subsumed within another 
result $r_2$ ($r_1 \sqsubset r_2$), if for every non-null variable binding in 
$r_1$, $r_2$ has the same binding, and $r_2$ has more non-null variable 
bindings 
than $r_1$. Thus results 3--5 in Res2 are subsumed within result 2. The 
\textit{best-match} operator (or \textit{mimumum union} as defined by 
Galindo-Legaria in \cite{galindosigmod}) removes all the subsumed results (see 
Res3). Final results of the query are given as 
\textit{best-match(nullification($(T_1 \leftouterjoin T_2) \leftouterjoin 
T_3$))}.

\subsection{Nullification, Best-match, and Minimality of triples} 
\label{sec:nullmin}
In \cite{atresigmod15}, we made an important observation that if every triple 
pattern in an OPTIONAL pattern query has \textit{minimal} triples associated 
with it, then nullification and best-match operations are not required (ref. 
Lemma 3.1 in \cite{atresigmod15}). We made yet another observation through 
Lemmas 3.3 and 3.4 in \cite{atresigmod15}, that for the following two classes 
of the OPTIONAL pattern queries nullification and best-match is not required if 
we use Algorithm-1 in \cite{atresigmod15} to reduce the set of triples 
associated with each triple pattern in the query, prior to generating final 
results using \textit{multi-way pipeline join} algorithm (Algorithm-3 in 
\cite{atresigmod15}).
\begin{itemize}
 \item \textit{Acyclic GoT}: OPTIONAL pattern queries whose GoT is acyclic.
 \item \textit{Only one join variable per slave}: OPTIONAL pattern queries 
where there are cycles in the GoT, but any slave triple pattern has only one 
join variable in it, and that join variable is shared with its master triple 
pattern.
\end{itemize}

These classes of OPTIONAL pattern queries are considered \textit{good} because 
they can avoid the overheads of the nullification and best-match operations 
despite the reordering of inner and left-outer joins.
The important premise of these observations is that Algorithm-1 reduces the 
triples associated with each triple pattern in the OPTIONAL pattern query in 
such a way that even if we reorder the inner and left-outer joins while doing 
the multi-way pipelined join, it does not generate spurious results, and thus 
avoids the necessity of nullification and best-match.
We extend this class of \textit{good} OPTIONAL pattern cyclic queries beyond 
the ones in which all slaves have only one join variable, and these are 
discussed in Section \ref{sec:cyclic}. Before that we revisit our pruning 
method and multi-way joins to work with just GoT, and obviate the need of 
\textit{graph of join variables} (GoJ) that was used in our previous work 
\cite{atresigmod15}.

\section{Pruning Triples} \label{sec:prune}
SPARQL (in turn SQL) queries can be evaluated using different equivalent plans. 
All the plans output exact same results. Typically, a plan with the least cost 
is chosen for evaluation. In the previous sections, we 
established the relationship between structural properties of a BGP-OPT query, 
minimality of triples, and nullification, best-match operations. However, we get 
the benefit of avoiding nullification and best-match, only if the tuples 
associated with the query \textit{before} performing the reordered joins are in 
the minimal (or favorable) form. We ensure that by using the \textit{pruning} 
step before performing joins. The \textit{pruning} phase only prunes the 
triples associated with each triple pattern in the query using a series of 
\textit{semi-joins}, and the \textit{multi-way pipelined joins} then produce 
the join results in a pipelined fashion, followed by nullification and 
best-match if required. 

\subsection{Semi-joins} \label{sec:semijoins}
Pruning of triples without performing joins is achieved through 
\textit{semi-joins}.
Semi-joins can be notationally represented as follows. $T_2 \ltimes_{?j} T_1 = 
\{t\ |\ t \in T_2, t.?j \in (\pi_{?j}(T_1) \cap \pi_{?j}(T_2))\}$. Here $t$ is a 
triple matching $T_2$, and $t.?j$ is a variable binding (value) of variable $?j$ 
in $t$. After this semi-join, $T_2$ is left with only triples whose $?j$ 
bindings are also in $T_1$, and all other triples are removed.

Bernstein et al \cite{semij1,semij2} and Ullman \cite{ullman} have proved 
previously that if the \textit{graph of tables} (GoT) of an inner-join query 
is \textit{acyclic}, a bottom-up followed by a top-down pass with 
\textit{semi-joins} at each table in this tree, reduces the set of 
tuples in each table to a minimal. Note that for the discussion of minimality 
of triples and semi-joins, we focus on the GoT of a query, and not GoSN. GoSN 
inherently encapsulates GoT inside it, since GoSN maintains connections between 
BGP blocks, whereas GoT maintains connections between individual triple 
patterns.

In our previous work, we proposed a pruning algorithm for BGP-OPT queries that 
makes use of \textit{graph of join variables} (GoJ) and 
\textit{clustered-semi-joins}. However, in this paper we propose an improved 
algorithm. Our algorithm is reminiscent of the \textit{full reducer} semi-join 
sequence as given in \cite{ullman}, that uses the concept of graphs with 
\textit{hyperedges} to represent tables (triple patterns for SPARQL). But 
full-reducers only addressed the inner-joins, and we address an intermix of 
inner and left-outer joins. We first discuss it in Algorithm-\ref{alg:prune}, 
and then discuss its differences from our previous Algorithm-1 in 
\cite{atresigmod15}.

\begin{algorithm}[h]
\small{
\SetKw{return}{return}
\SetKwFunction{greedy}{greedy-semij-order}
\SetKwFunction{ordsn}{order-supernodes}
\SetKwFunction{tpsn}{get-tps-in-SN}
\SetKwFunction{leaves}{one-eqv-class}
\SetKwFunction{mincost}{mincost}
\SetKwFunction{semij}{semi-join}
\SetKwInOut{Input}{input}
\Input{GoSN}
sn-order = \ordsn{GoSN}\; \label{ln:orderslaves}
\If{GoT cyclic AND cycles in slaves}{ \label{ln:cyclic}
$order_{greedy}$ = \greedy{sn-order}\;\label{ln:greedyord}
  \For{each $T_i \ltimes T_p$ in $order_{greedy}$}{
    \semij($T_i$, $T_j$)\tcp*{Alg \ref{alg:semij}}
  }
\return
}\label{ln:endcyclic}
\ElseIf{GoT cyclic AND only $SN_{abs}$ cyclic}{ \label{ln:partcycle}
$order_{greedy}$ = \greedy{$SN_{abs}$}\;
  \For{each $T_i \ltimes T_p$ in $order_{greedy}$}{
    \semij($T_i$, $T_j$)\tcp*{Alg \ref{alg:semij}}
  }
  sn-order.remove($SN_{abs}$)\; \label{ln:partcycleend}
} 
\For{each supernode $SN_i$ in sn-order}{\label{ln:eachsn}
  list \textit{tp-sn} = \tpsn{$SN_i$}\;
  \For{each $T_i$ in tp-sn}{ \label{ln:eachtp}
    Create equiv classes of edges incident on $T_i$ in $SN_i$\;
  }
  \While{tp-sn not empty}{
    list \textit{tp-leaves} = \leaves{$SN_i$}\; \label{ln:leaves}
    $T_i$ = \mincost{tp-leaves}\; \label{ln:mincost}
    \If{$SN_i$ is slave}{ \label{ln:ifslave}
      \For{each $T_m$ master TP of $T_i$}{
      \tcp{$T_m$ and $T_i$ share a join variable}
	$order_{bu}$.append($T_i \ltimes T_m$)\; \label{ln:transfer}
      }
    }\label{ln:endifslave}
    $T_j$ = mincost neighbor of $T_i$\; \label{ln:mincostnb}
    $order_{bu}$.append($T_j \ltimes T_i$)\; \label{ln:append}
    remove $T_i$ and its edges from consideration\;
  }
  \For{each $T_i \ltimes T_p$ in $order_{bu}$}{
    \semij($T_i$, $T_j$)\tcp*{Alg \ref{alg:semij}}
  }
  $order_{td}$ = $order_{bu}$.reverse()\; \label{ln:td}
  Remove semi-joins $T_m \ltimes T_i$ from $order_{td}$\; \label{ln:remove}
  \For{each $T_i \ltimes T_p$ in $order_{td}$}{
    \semij($T_i$, $T_j$)\tcp*{Alg \ref{alg:semij}}
  }
}\label{ln:eachsnend}
\return
}
\caption{\texttt{prune\_triples}}\label{alg:prune}
\end{algorithm}

\ref{alg:prune}. 
The working of the algorithm is described as follows. We first order all the 
supernodes in the GoSN of a query according to the master-slave hierarchy -- 
masters before their respective slaves, and among any two supernodes not in 
the master-slave hierarchy, we randomly pick one over another. Note that since 
these are \textit{well-designed} queries and we have coalesced all the absolute 
master supernodes into $SN_{abs}$, the relative ordering among two such 
supernodes not in the master-slave hierarchy does not matter. We call this 
ordering among the supernodes \texttt{sn-order} as given by the 
\texttt{order-supernodes} function (ln \ref{ln:orderslaves}), and $SN_{abs}$ 
appears first in this order.
If GoT is cyclic and there are cycles in slave supernodes too, then we just do 
pruning by following a greedy order of semi-joins  -- doing semi-joins over 
highly selective triple patterns first -- honoring the master-slave hierarchy 
among the triple patterns (\ref{ln:cyclic}--\ref{ln:endcyclic}), and end the 
pruning process there. In this case nullification and best-match are necessary 
after \texttt{multi-way-joins}. However, if GoT is cyclic but the cycles are 
confined only to $SN_{abs}$, and GoTs of slaves are acyclic, we consider a 
greedy order of semi-joins over the triple patterns in $SN_{abs}$, prune 
the triples in $SN_{abs}$, and remove $SN_{abs}$ from \texttt{sn-order} (ln 
\ref{ln:partcycle}--\ref{ln:partcycleend}).

\fxnote{Give a proper definition of eqv classes.}
Then, starting with the next supernode, $SN_i$, in \texttt{sn-order} (ln 
\ref{ln:eachsn}) we consider all the triple patterns encapsulated within that 
supernode.
For each triple pattern $T_i$ in $SN_i$, and its connected triple patterns in 
$SN_i$ alone, we create \textit{equivalence classes} of edges of GoT incident on 
$T_i$ (recall the definition of \textit{equivalence class} of edges given in 
Section \ref{sec:prelim}). We do not consider edges that connect $T_i$ with 
triple patterns outside $SN_{i}$. After doing this for all the triple patterns 
in $SN_i$, we pick the triple patterns that have only \textit{one equivalence 
class} among its edges (ln \ref{ln:leaves}). This triple pattern is a 
\textit{leaf node}. Among all such leaf nodes, we pick the one which has the 
\textit{least} number of triples associated with it -- most \textit{selective} 
one. In case of a tie, we pick one randomly (ln \ref{ln:mincost}).

Among leaf $T_i$'s connected triple patterns within $SN_i$, we pick the 
neighbor with the least number of triples associated with it, say $T_j$ (ln 
\ref{ln:mincostnb}), and add a semi-join $T_j \ltimes T_i$ to the queue 
$order_{bu}$ (bottom-up semi-join order) (ln \ref{ln:append}).
If $SN_i$ is a slave supernode, we ensure to fetch its master's variable binding 
restrictions as follows. If $SN_b$ is a master of $SN_i$, such that $SN_b 
\rightarrow SN_i$, we \textit{transfer the constraints on the variable bindings} 
from $SN_b$ to $SN_i$ as follows. Without losing generality, while adding a 
semi-join between two triple patterns $T_j \ltimes T_i$ in $SN_i$ to 
$order_{bu}$, we first check if $T_i$ has a neighbor $T_m$ in $SN_b$. If it 
does, then we add $T_i \ltimes T_m$ to $order_{bu}$ before $T_j \ltimes T_i$. 
This ensures that any variable bindings imposed by a master of the triple 
patterns are transferred to their respective slave triple patterns (ln 
\ref{ln:transfer}). Next we remove $T_i$ and all the edges incident on it from 
consideration, and repeat the same procedure described above with the rest of 
the triple patterns in $SN_{i}$.

Notice that through this procedure we recursively define a spanning tree over 
the graph of triple patterns (GoT) encapsulated within $SN_{i}$, and make a 
bottom-up pass over it -- every semi-join $T_j \ltimes T_i$ denotes $T_j$ to be 
the parent of $T_i$ in the spanning tree, and the semi-join denotes the 
\textit{direction of the walk} from the leaves to the internal nodes and the 
root of the tree, e.g., semi-join $T_j \ltimes T_i$ denotes a walk from $T_i 
\rightarrow T_j$ on the spanning tree. For the top-down pass, we simply reverse 
the queue $order_{bu}$ to get $order_{td}$ (ln \ref{ln:td}). That is for every 
semi-join of type $T_j \ltimes T_i$ in $order_{bu}$, we add $T_i \ltimes T_j$ 
to $order_{td}$. However, we omit all the semi-joins of type $T_m \ltimes T_i$  
where $T_m$ is a master of $T_i$ (ln \ref{ln:remove}).
When the GoT is acyclic, the very first $SN_i$ is always $SN_{abs}$.

The main differences between our previously proposed technique (ref Algorithms 
3.1, 3.2 in \cite{atresigmod15}), and Algorithm \ref{alg:prune} here are as 
follows:
\begin{itemize}[label=$\bullet$]
 \item In the present technique, we form a rooted tree over GoT, whereas in 
\cite{atresigmod15}, the rooted tree was formed over \textit{graph of join 
variables} (GoJ).
 \item In our present technique the rooted tree over GoT is formed in a 
\textit{bottom-up} fashion, where we first pick the \textit{most selective} 
triple patterns as the leaves and recursively build the internal nodes and root 
of the tree. In 
\cite{atresigmod15}, the rooted tree over GoJ was formed in a \textit{top-down} 
fashion, where we first picked a join variable with \textit{lowest selectivity} 
as the root of the tree, and then recursively picked the internal nodes and 
leaves.
 \item In \cite{atresigmod15}, we first did a bottom-up pass over \textit{all} 
the induced GoJs of individual supernodes, and then did the top-down passes.  
In the present technique, we do the bottom-up and top-down passes on the triple 
patterns of each supernode in one go. Thus we do the full possible pruning of 
the triples associated with the triple patterns in a supernode, before moving on 
to its slaves. This helps in better pruning of the slaves, because their 
respective masters are already in the pruned state.
\end{itemize}

Our discussion of the properties of nullification, cyclicity of queries, 
minimality of triples, and the pruning procedure until now has been agnostic to 
the lower level storage structure and indexes on the RDF graph. Indeed, our 
propositions and algorithms presented earlier can work on 
any storage and index structure -- only the \textit{semi-joins} procedure used 
in Algorithm-\ref{alg:prune} will change depending on the storage structure and 
indexes. Nevertheless, an \textit{efficient} storage and index structure helps 
in achieving better performance. We achieve this through the usage of compressed 
bitvector indexes on RDF graphs, and procedures that directly work on these 
compressed indexes without decompressing them as proposed in 
\cite{bitmatwww10,atresigmod15}. For the completeness of the text, these are 
described in Section \ref{sec:index}.

\section{BitMat Indexes}\label{sec:index}

In an RDF graph, let $V_s$, $V_p$, and $V_o$ be the sets of unique subject, 
predicate, and object values. Then a 3D bitcube of RDF data has $V_s \times V_p 
\times V_o$ dimensions. Each cell in this bitcube represents a unique RDF 
triple formed by the coordinate values (S P O). If this (S P O) triple is 
present in the given RDF dataset, that bit is set to 1 in the bitcube.
The unique values of subjects, predicates, and objects in the original RDF 
data are first mapped to integer IDs, which in turn are mapped to the bitcube 
dimensions. To facilitate joins on S-O dimensions, same S and O values are 
mapped to the same coordinates of the respective 
dimensions\footnote{\scriptsize{For the scope of this paper, we do not consider 
joins on S-P or O-P dimensions.}}.

Let $V_{so} = V_s \cap V_o$. Set $V_{so}$ is mapped to a sequence of integers 1 
to $|V_{so}|$. Set $V_s - V_{so}$ is mapped to a sequence of integers $|V_{so}| 
+ 1$ to $|V_s|$. Set $V_o - V_{so}$ is mapped to a sequence of integers 
$|V_{so}| + 1$ to $|V_o|$, and set $V_p$ is mapped to a sequence of integers 1 
to $|V_p|$.

This bitcube is conceptually sliced along each dimension, and the 2D 
\textit{BitMats} are created. In general, four types of 2D BitMats are created: 
(1) S-O and O-S BitMats by  slicing the P-dimension (O-S BitMats are nothing 
but a \textit{transpose} of the respective S-O BitMats), (2) P-O BitMats by 
slicing the S-dimension, and (3) P-S BitMats by slicing the O-dimension. 
Altogether we store $2*|V_p| + |V_s| + |V_o|$ BitMats for any RDF data. Figure 
\ref{fig:bitcube} shows 2D S-O BitMats that we can get by slicing the predicate 
dimension (others are not shown for conciseness).

Each row of these 2D BitMats is compressed as follows. In the 
run-length-encoding, a bit-row like ``11100-11110'' is represented as ``[1] 3 2 
4 1'', and  ``0010010000'' is represented as ``[0] 2 1 2 1 4''. Notably, in the 
second case, the bit-row has only two set bits, but it has to use \textit{five} 
integers in the compressed representation. So we use a \textit{hybrid} 
representation in our implementation that works as follows -- if the number of 
set bits in a bit-row are less than the number of integers used to represent it, 
then we simply store the set bit positions. So ``0010010000'' will be compressed 
as ``3 6'' (3 and 6 being the positions of the set bits). This hybrid 
compression fetches us as much as 40\% reduction in the index space compared to 
using only run-length-encoding as done in \cite{bitmatsrc}.

Other meta-information such as, the number of triples, and condensed 
representation of all the non-empty rows and columns in each BitMat, is also 
stored along with each BitMat. This information helps us in quickly determining 
the number of triples in each BitMat and its selectivity without counting each 
triple in it, while processing the queries. A 2D S-O or O-S BitMat of predicate 
\textit{:hasFriend} represents all the triples matching a triple pattern of kind 
(\textit{?a :hasFriend ?b}), a 2D P-S BitMat of O-value \textit{:Seinfeld} 
represents all triples matching triple pattern (\textit{?c ?d :Seinfeld}), and 
so on.

\begin{figure}[h]
    \centering
	\includegraphics[scale=0.265]{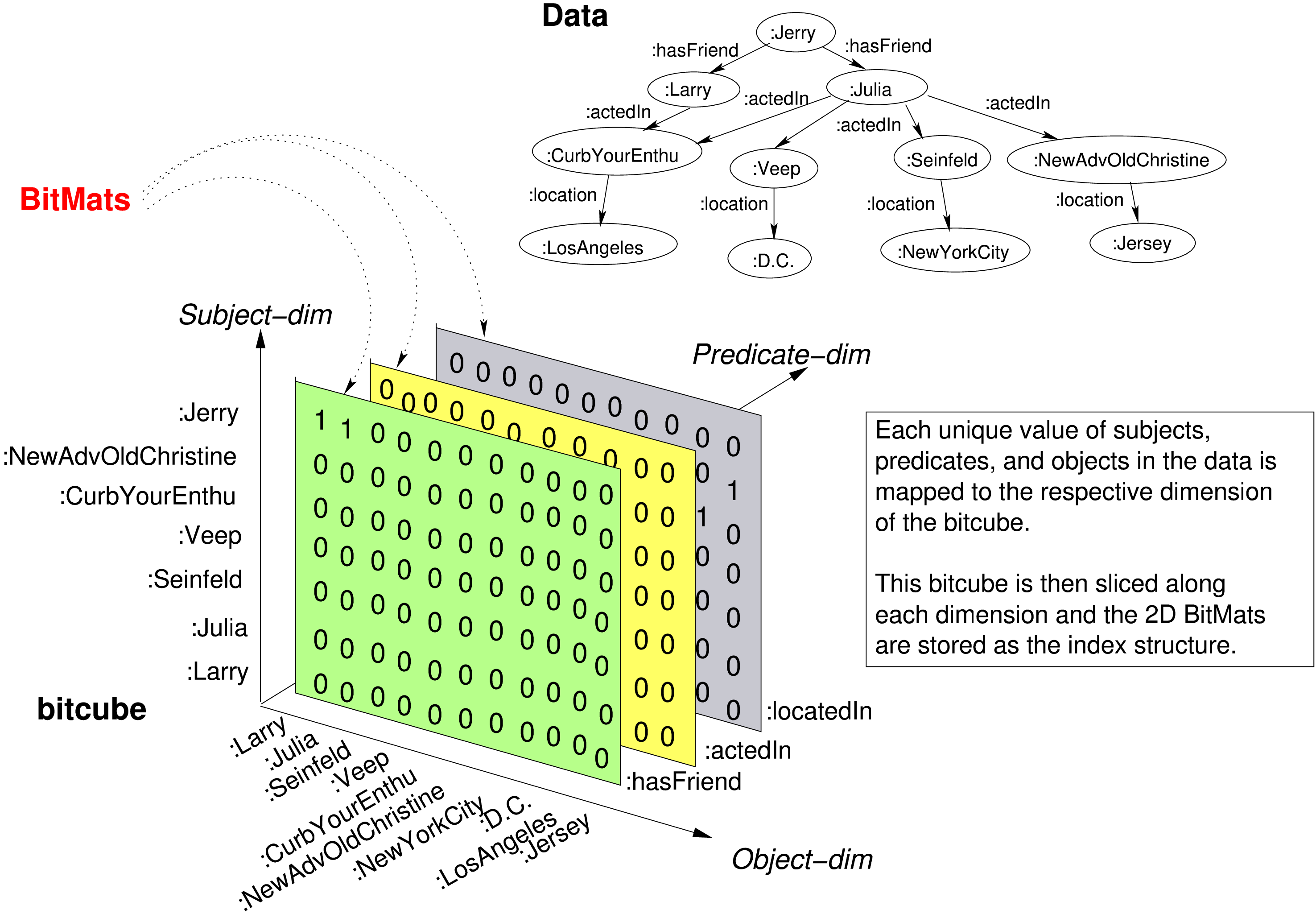}
        \caption{3D Bitcube of RDF data in Figure \ref{fig:nullbest}}
    \label{fig:bitcube}
\end{figure}

Query execution uses the \textit{fold} and \textit{unfold} primitives, which 
process the compressed BitMats without uncompressing them \cite{bitmatwww10}.

\textbf{\textit{Fold}} operation is represented as `\texttt{fold(BitMat, 
Re-\\tainDimension) returns bitArray}'. It takes a 2D BitMat and 
\textit{folds} it by retaining the \texttt{RetainDimension}.
More succinctly, a fold operation is nothing but \textit{projection} of 
the distinct values of the particular BitMat dimension, by doing a bitwise OR 
on the other dimension. It can be represented as:
\[
\mathtt{fold(BM_{T_i}, dim_{?j})} \equiv \pi_{?j}(BM_{T_i})
\]

$BM_{T_i}$ is a 2D BitMat holding the triples matching $T_i$, and $dim_{?j}$ is 
the dimension of BitMat that represents variable $?j$ in $T_i$.
E.g., for a triple pattern (\textit{?friend :actedIn ?sitcom}), if we consider 
the O-S BitMat of predicate \textit{:actedIn}, \textit{?friend} values are in 
the \textit{column} dimension, and \textit{?sitcom} values are in the 
\textit{row} dimension of the BitMat.

\textbf{\textit{Unfold}} is represented as `\texttt{unfold(BitMat, 
MaskBit\\Array, RetainDimension)}'. For every bit set to 0 in the 
\texttt{MaskBitArray}, unfold clears all the bits corresponding to that  
position of the \texttt{RetainDimension} of the BitMat. Unfold  can be simply 
represented as:
\[
\mathtt{unfold(BM_{T_i}, \beta_{?j}, dim_{?j})} \equiv \{t\ |\ t \in BM_{T_i}, 
t.?j \in \beta_{?j}\}
\]

$t$ is a triple in $BM_{T_i}$ that matches $T_i$. $\beta_{?j}$ is the 
\texttt{MaskBitArray} containing bindings of $?j$ to be retained. $dim_{?j}$ is 
the dimension of $BM_{T_i}$ that represents $?j$, and $t.?j$ is a binding of 
$?j$ in triple $t$. In short, \texttt{unfold} keeps only those triples whose 
respective bindings of $?j$ are set to 1 in $\beta_{?j}$, and removes all other.

Until now we discussed semi-joins only notationally as $T_i \ltimes_{?j} T_k 
= \{t\ |\ t \in T_i, t.?j \in \{\pi_{?j}(T_i) \cap \pi_{?j}(T_k)\}\}$. 
Semi-joins can be implemented using the BitMat indexes and the \texttt{fold, 
unfold} primitives as given in Algorithm \ref{alg:semij}.

\begin{algorithm}[h]
\small{
\SetKwFunction{fold}{fold}
\SetKwFunction{unfold}{unfold}
\SetKw{return}{return}
\SetKwInOut{Input}{input}
\Input{$T_i$, $T_k$, $?j$}
\BlankLine
\tcp{$T_i \ltimes T_k$}
$\beta_?j$ = \fold{$BM_{T_i}$, $dim_{?j}$} AND \fold{$BM_{T_k}$, 
$dim_{?j}$}\; 
\label{ln:collect}
\unfold{$BM_{T_i}$, $\beta_{?j}$, $dim_{?j}$}\; \label{ln:distr}
}
\caption{\texttt{semi-join}}\label{alg:semij}
\end{algorithm}

\section{Multi-way Pipelined Joins}
In the previous sections, we established the algorithmic basis of our techniques 
for understanding the characteristics of a BGP-OPT query and pruning the 
triples associated with the each 
triple pattern in the query \textit{before} generating the final results -- 
Algorithm \ref{alg:prune} only \textit{prunes} the candidate RDF triples, 
but does not generate the final results. In Section \ref{sec:index}, we took an 
overview of our storage and indexing structure for the RDF graphs. In this 
section, put the pruning algorithm and index structure together with our 
technique of \textit{multi-way-pipelined joins} to generate the final results 
as given in Algorithm-\ref{alg:qproc}. 

\begin{algorithm}[h]
\small{
\SetKwFunction{getgosn}{get-GoSN}
\SetKwFunction{getgoj}{get-graph-jvars}
\SetKwFunction{getjord}{get\_jvar\_order}
\SetKwFunction{bestmatch}{decide-best-match-reqd}
\SetKwFunction{init}{init}
\SetKwFunction{prune}{prune\_triples}
\SetKwFunction{bestmop}{best-match}
\SetKwFunction{finalres}{multi-way-join}
\SetKwFunction{sorttps}{spanning-tree-tps}
\SetKwFunction{tporder}{sort-tps-master-slave}
\SetKw{return}{return}
\SetKwInOut{Input}{input}
\SetKwInOut{Output}{output}
\Input{Original BGP-OPT query}
\Output{Final results}
\BlankLine
GoSN = \getgosn{Orig BGP-OPT query}\; \label{ln:getgosn}
\tcp{Based on cyclicity}
\textbf{bool} \textit{NB-reqd} = \bestmatch{GoSN, GoJ}\; \label{ln:bestmdec}
\prune{GoSN}\tcp*{Alg \ref{alg:prune}} \label{ln:callprune}
\BlankLine
tporder = \tporder()\; \label{ln:tporder}
stps = \sorttps()\; \label{ln:sorttps}
vmap = empty-map, size as num of vars in query\;
\tcp{Final result generation - Alg \ref{alg:multiway}} 
\textit{allres} = \finalres{vmap, stps, visited, NB-reqd}\; 
\label{ln:finalres}
\BlankLine
\eIf{NB-reqd}{ 
  \textit{finalres} = \bestmop{allres}\; \label{ln:bestmatch}
}{
  \textit{finalres} = \textit{allres};
}
\return{finalres}\;
}
\caption{Query processing}\label{alg:qproc}
\end{algorithm}

In Algorithm \ref{alg:qproc} for query processing, we first construct the GoSN 
(ln \ref{ln:getgosn}). Next, we decide if nullification and best-match are 
required. This decision is based on the cyclic properties the query, and we 
discuss them in Section \ref{sec:cycleprop}, and Lemmas \ref{lemma:acyclic} and 
\ref{lemma:cyclic}. We call \texttt{prune\_triples} (Algorithm-\ref{alg:prune}) 
to prune the unwanted triples. Pruning operation also on-the-fly loads the 
BitMat associated with each triple pattern containing only triples satisfying 
that triple pattern (we have not shown this operation explicitly in 
Algorithm-\ref{alg:prune}, but is explained next). We choose an appropriate 
BitMat for each triple pattern as follows. If the triple pattern in the query is 
of type (\textit{?var :fx1 :fx2}), i.e., with two fixed positions, we load only 
one row corresponding to \textit{:fx1} from the P-S BitMat for \textit{:fx2}. 
Similarly for a TP of type (\textit{:fx1 :fx2 ?var}), we load only one row 
corresponding to \textit{:fx2} from the P-O BitMat for \textit{:fx1}. E.g., for 
(\textit{?sitcom :location :NYC}) we load only one row corresponding to  
\textit{:location} from the P-S BitMat of \textit{:NYC}. If the triple pattern 
is of type (\textit{?var1 :fx1 ?var2}), we load either the S-O or O-S BitMat of 
\textit{:fx1}. If \textit{?var1} is a join variable and \textit{?var2} is not, 
we load the S-O BitMat and vice versa. If both, \textit{?var1} and 
\textit{?var2}, are join variables, we observe the order of semi-joins in 
$order_{bu}$ to check if the semi-join over the given triple pattern is over 
\textit{?var1} or \textit{?var2} first. If semi-join over \textit{?var1} 
comes before \textit{?var2}, we load the S-O BitMat and vice versa.

While loading the BitMats, we do active pruning using the triple patterns whose 
BitMats are already initialized. E.g., if we first load BitMat $BM_{T_1}$ 
containing triples matching (\textit{:Jerry :hasFriend ?friend}), then 
while loading $BM_{T_2}$, we use the bindings of \textit{?friend}
in $BM_{T_1}$ to actively prune the triples in $BM_{T_2}$ while loading it. 
Then while loading $BM_{T_3}$, we use the bindings of \textit{?sitcom} in 
$BM_{T_2}$ to actively prune the  triples in $BM_{T_3}$.
We check whether two triple patterns are joining with each other over an inner 
or left-outer join using GoSN with the \textit{master-slave} or \textit{peer} 
relationship, and then we decide whether to use other BitMat's variable 
bindings.

Note that using \texttt{prune\_triples}, we prune the triples in BitMats, but 
we need to actually ``join'' them to produce the final results. For that we use 
\texttt{multi-way-join} (ln \ref{ln:finalres} in Algorithm \ref{alg:qproc}).
This procedure is described separately in Section \ref{sec:finalres}.
After \texttt{multi-way-join}, we use \texttt{best-match} to remove any 
subsumed results if nullification was required as a part of 
\texttt{multi-way-join} (discussed in Section \ref{sec:cycleprop}). In 
\texttt{best-match}, we externally sort all the results generated by 
\texttt{multi-way-join}, and then remove the subsumed results with a single pass 
over them.

\subsection{Multi-way Pipelined Join} \label{sec:finalres}
\begin{algorithm}[t]
\small{
\SetKw{return}{return}
\SetKw{continue}{continue}
\SetKw{true}{true}
\SetKw{false}{false}
\SetKwFunction{finalres}{multi-way-join}
\SetKwFunction{output}{output}
\SetKwFunction{masterof}{master-of}
\SetKwFunction{peerof}{peer-of}
\SetKwFunction{nuli}{nullification}
\SetKwInOut{Input}{input}
\SetKwInOut{Output}{output}
\Input{vmap, stps, visited, nulreqd}
\Output{all the results of the query}
\BlankLine
\If{visited.size == stps.size}{\label{ln:outputres}
  \If{nulreqd}{
    \nuli(vmap)\; \label{ln:nullification}
  }
  \output(vmap)\tcp*{generate a single result} \label{ln:endoutputres}
\return\;
}
\eIf{visited is empty}{ \label{ln:firsttp}
  $T_1$ = first TP from stps\;
  visited.add($T_1$)\;
  \For{each triple $t \in BM_{T_1}$}{
    generate bindings for vars($T_1$) from $t$, store in \texttt{vmap}\;
    \finalres{vmap, stps, visited, nulreqd}\;
  } \label{ln:endfirsttp}
}{
$T_i$ = next triple pattern in \texttt{stps}\;

\BlankLine
    \textit{atleast-one-triple} = \false\;
     \For{$t \in BM_{T_i}$ with same bindings} {\label{ln:mapother}
      \textit{atleast-one-triple} = \true\;
      store vars($T_i$) bindings of $t$ in \texttt{vmap}\;
      visited.add($T_i$)\;
      \finalres{vmap, stps, visited, nulreqd}\;
      visited.remove($T_i$)\; \label{ln:endmapother}
    }
\BlankLine
  \If{(\textit{atleast-one-triple} == \false)}{
      \If{$T_i$ is an absolute master}{
      \return\; \label{ln:rollback}
      }
     \tcp{This means $T_i$ is a slave}
      set all vars($T_i$) to NULL in vmap\; \label{ln:nomap}
      visited.add($T_i$)\;
      \finalres{vmap, stps, visited, nulreqd}\;
      visited.remove($T_i$)\;\label{ln:endnomap}
  }
}
}
\caption{\texttt{multi-way-join}}\label{alg:multiway}
\end{algorithm}

Before calling \texttt{multi-way-join}, we first sort all the triple patterns 
in the query as follows. Considering the triple patterns in $SN_{abs}$, we sort 
them in the ascending order of the number of triples left in each triple 
pattern's BitMat. Then we sort the remaining supernodes in the descending order 
of master-slave hierarchy. That is, among two supernodes connected as $SN_1 
\rightarrow SN_2$, triple patterns in $SN_1$ are sorted before those in $SN_2$. 
Among the triple patterns in the same supernode (\textit{peers}), they are 
sorted in the ascending order of the number of triples left in their BitMats. 
Let us call this order \texttt{tporder} (ln \ref{ln:tporder} in Algorithm 
\ref{alg:qproc}). From \texttt{tporder}, we construct a conceptual 
\textit{spanning tree} as follows. The very first triple pattern is 
designated as the \textit{root} of the tree, and added to a new sort order of 
triple patterns, called \texttt{stps}. We recursively pick the next triple 
pattern from \texttt{tporder} such that it is connected to at least one triple 
pattern in \texttt{stps} (ln \ref{ln:sorttps} in Algorithm \ref{alg:qproc}).
This \texttt{stps} is used in \texttt{multi-way-join} to produce final results 
and decide the join order. In \texttt{multi-way-join} we also use at most 
$\sum_{T_i \in Q} vars(T_i)$ additional memory buffer, where $vars(T_i)$ are the 
variables in every triple pattern $T_i$ in the query $Q$. This is \texttt{vmap} 
in Alg \ref{alg:multiway}. Thus we use negligible additional memory in 
\texttt{multi-way-join}.

At the beginning, \texttt{multi-way-join} gets \texttt{stps}, an empty 
\texttt{vmap} for storing the variable bindings, an empty \texttt{visited} 
list, and a flag \texttt{nulreqd} indicating if nullification is required. We 
go over each triple in $BM_{T_1}$ of the first triple pattern in \texttt{stps}, 
generate bindings for the variables in $T_1$, and store them in \texttt{vmap}. 
We add $T_1$ to the \texttt{visited} list,
and call \texttt{multi-way-join} recursively for the rest of the triple 
patterns (ln \ref{ln:firsttp}--\ref{ln:endfirsttp}). In each recursive call, 
\texttt{multi-way-join} gets a partially populated \texttt{vmap} and a 
\texttt{visited} list that tells which triple pattern's variable bindings are 
already stored in \texttt{vmap}. \texttt{Stps} order is formed in a way that 
from second position onward each triple pattern in \texttt{stps} has at least 
one connection in \texttt{stps} order before it. Thus, for the next recursive 
call of \texttt{multi-way-join}, we are expected to find at least one variable 
binding in \texttt{vmap} for the next non-visited $T_i$ (ln \ref{ln:getbind}).
\texttt{Stps} order ensures that a master triple pattern's variable bindings are 
stored in \texttt{vmap} before its slaves. If there 
exists one or more triples $t$ in $BM_{T_i}$ consistent with the variable 
bindings in \texttt{vmap}, then for each such $t$ we generate bindings for all 
the variables in $T_i$, store them in \texttt{vmap}, and proceed with the 
recursive call to \texttt{multi-way-join} for the rest of the triple patterns 
(ln \ref{ln:mapother}--\ref{ln:endmapother}). Notice that, this way we 
\textit{pipeline} all the BitMats, and do not use any other intermediate storage 
like hash-tables.

If we do not find any triple $t$ in $BM_{T_i}$ consistent with the existing 
variable bindings in \texttt{vmap}, then -- (1) if $T_i$ is an absolute master, 
we \textit{rollback} from this point, because an absolute master triple pattern
cannot have NULL bindings (ln \ref{ln:rollback}), else (2) we map all the 
variables in $T_i$ to NULLs, and proceed with the recursive call to 
\texttt{multi-way-join} (ln \ref{ln:nomap}--\ref{ln:endnomap}).
When all the triple patterns in the query are in the \texttt{visited} list, we 
check if we require \texttt{nullification} to ensure consistent variable 
bindings in \texttt{vmap} across all the slave triple patterns, and output one 
result (ln \ref{ln:outputres}--\ref{ln:endoutputres}).
We continue this recursive procedure till triples in $BM_{T_1}$ are exhausted
(ln \ref{ln:firsttp}--\ref{ln:endfirsttp}).
Intuitively, \texttt{multi-way-join} is reminiscent of a relational join plan
with \textit{reordered left-outer-joins} -- that is, in \texttt{stps} we sort 
master triple patterns before their slaves, and masters generate variable 
bindings before slaves in \texttt{vmap}.

\section{Cycles in the Queries} \label{sec:cycleprop}

In this section we discuss the cyclic properties of the queries and how they 
affect the requirement of nullification and best-match operations after 
reordering the inner and left-outer-joins in a query, that makes an important 
part of reorderability of joins.

\subsection{Acyclic Queries} \label{sec:pruneacyclic}
In \cite{atresigmod15}, we proposed our pruning algorithm for OPTIONAL pattern 
queries by constructing a \textit{graph of join variables} (GoJ) as follows. 
Every join variable in the query is a node in GoJ, and two join variables have 
an undirected edge between them, if both the join variables co-occur in a 
triple pattern in the query. In \cite{atresigmod15}, we proposed that if the GoT 
of an OPT query is acyclic then the GoJ is acyclic too. However, there exists 
queries where this equivalence may not hold. An example of such a query is given 
below.

\vspace{2mm}
\begin{lgrind}
{\small
\L{\LB{\K{SELECT}_?\V{a}_?\V{b}_?\V{c}}}
\L{\LB{\K{WHERE}_\{}}
\L{\LB{}\Tab{4}{?\V{a}_:\V{p1}_?\V{b}_.}}
\L{\LB{}\Tab{4}{?\V{b}_:\V{p2}_?\V{c}_.}}
\L{\LB{}\Tab{4}{?\V{c}_:\V{p3}_?\V{a}_.}}
\L{\LB{}\Tab{4}{?\V{a}_?\V{b}_?\V{c}_.}}
\L{\LB{\}}}
}
\end{lgrind}
\vspace{2mm}

\begin{figure}
 \centering
 \includegraphics[scale=0.45]{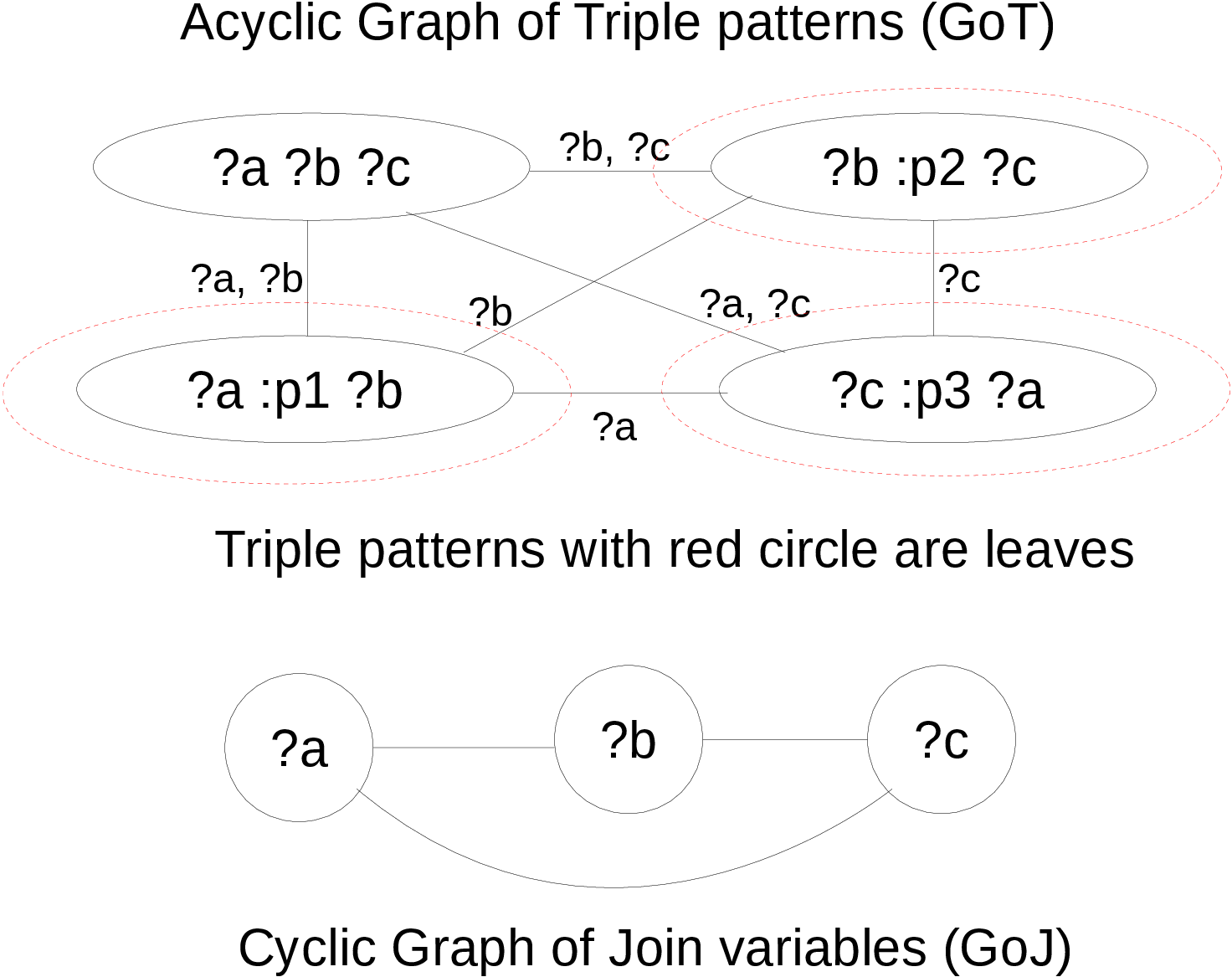}
 \caption{An acyclic query whose GoJ is cyclic} \label{fig:exception}
\end{figure}

\fxnote{draw a fig for this query}
Note that in this query the triple pattern \textit{(?a ?b ?c)} joins with every 
other triple pattern over two join variables. Per our definition of acyclicity 
given in Section \ref{sec:acyclic}, this query is 
indeed \textit{acyclic}. However, its GoJ is cyclic as shown in Figure 
\ref{fig:exception}. Also note that in this query there is a join on the 
\textit{object-predicate} position ($T_1$ and $T_3$) which is very uncommon in 
the RDF data. Most SPARQL queries have triple patterns joining over only one 
variable, and these joins are on subject-object, 
subject-subject, or object-object positions (since they naturally indicate 
edges incident on the nodes in the graph). Hence our original proposition and 
the technique of \textit{clustered-semi-joins} on GoJ  in \cite{atresigmod15}
still work correctly for the SPARQL queries where any two triple patterns always 
join only over one join variable (and all the queries used in our experiments 
in \cite{atresigmod15} satisfied this condition).
Nevertheless, our improved \texttt{prune\_triples} method (Algorithm 
\ref{alg:prune}) in this article works correctly for such corner case queries 
too.

\begin{lemma} \label{lemma:acyclic}
Nullification and best-match can be avoided for an OPTIONAL pattern query 
with an acyclic GoT.
\end{lemma}

\begin{proof}
Nullification and best-match processes are described in Section 
\ref{sec:nullbm}. In that, it can be observed that nullification is required 
when some variable in a triple pattern $T_i$ is bound to a value that does not 
exist in another triple pattern $T_j$ that is either $T_i's$ master or peer.
When a BGP-OPT query is completely acyclic -- individual GoTs of each OPT-free 
BGP component as well as the entire GoT of the query across all the triple 
patterns is acyclic -- \texttt{prune\_triples} (Algorithm \ref{alg:prune}) 
ensures that each triple pattern is left with minimal number of triples -- 
there does not exist any binding for a variable in $T_i$ that does not exist in 
its master or peer triple pattern $T_j$.

This minimality is ensured as follows. When we do pruning of triples in 
$SN_{abs}$, the triples associated with triple patterns in $SN_{abs}$ are 
reduced to minimal. This follows directly from the proofs of Bernstein et al. 
and Ullman \cite{semij1,semij2,ullman}. When we do pruning of triples in slave 
supernodes, the bindings in its acyclic master supernode are already 
minimalized, and we propagate those on the present slave supernode's 
variable bindings (lines \ref{ln:ifslave}--\ref{ln:endifslave} in 
Algorithm-\ref{alg:prune}). Thus after completion of \texttt{prune\\\_triples}, 
each triple pattern in the query has minimal triples associated with its triple 
patterns. Thus nullification and best-match can be avoided. \qed
\end{proof}

\subsection{Cyclic Queries} \label{sec:cyclic}
For OPT-free BGPs, i.e., pure inner-joins, with a cyclic GoT, minimality of 
triples cannot be guaranteed using the procedure of bottom-up followed by 
top-down pass of semi-joins on the spanning tree over GoT as described in in 
Section \ref{sec:pruneacyclic} \cite{semij1,semij2,ullman}. This result carries 
over immediately to the queries with an intermix of BGP and OPT patterns too 
whose GoT is cyclic. Thus, we simply prune the triples by following a 
\textit{greedy} order of semi-joins, while adhering to the master-slave 
hierarchy among the triple patterns in the query. The greedy order of 
semi-joins is determined by the relative selectivity of the triple patterns, 
and the master-slave hierarchy between them (ln \ref{ln:greedyord} in Algorithm 
\ref{alg:prune}). Since minimality of triples in each TP is not guaranteed, we 
need to use the \textit{nullification} and \textit{best-match} operations in a 
reordered query to ensure consistent variable bindings, and to remove subsumed 
results.

This observation in general holds for all cyclic BGP-OPT queries, but we 
identify a \textit{subclass} of cyclic BGP-OPT queries that can \textit{avoid} 
nullification and best-match -- in such queries the following conditions hold:
\begin{enumerate}
 \item A subgraph of GoT representing only the triple patterns in any slave 
supernode is acyclic.
  \item There is only \textit{equivalence class} of edges connecting triple 
patterns in the master-slave hierarchy between two supernodes. That is, if $SN_a 
\rightarrow SN_b$, and we put all GoT edges of type $\langle T_m, T_s\rangle, 
T_m \in SN_a, T_s \in SN_b$ in equivalence classes, then there is only one 
equivalence class. Generalizing, this property holds for each and every 
pair of supernodes $SN_x \rightarrow SN_y$ in master-slave hierarchy that have 
an directed edge among them in GoSN. Recall our definition of equivalence 
classes of edges given in Definition \ref{def:eqvclass}.
\end{enumerate}

For such queries, we do greedy way of pruning using semi-joins over the triple 
patterns in $SN_{abs}$, and after that for each slave supernode, 
we follow the same procedure of bottom-up and top-down pruning as described in 
Algorithm \ref{alg:prune} (ref. ln \ref{ln:partcycle}--\ref{ln:partcycleend} 
and ln \ref{ln:eachsn}--\ref{ln:eachsnend}).

\begin{figure}
 \centering
 \includegraphics[scale=0.4]{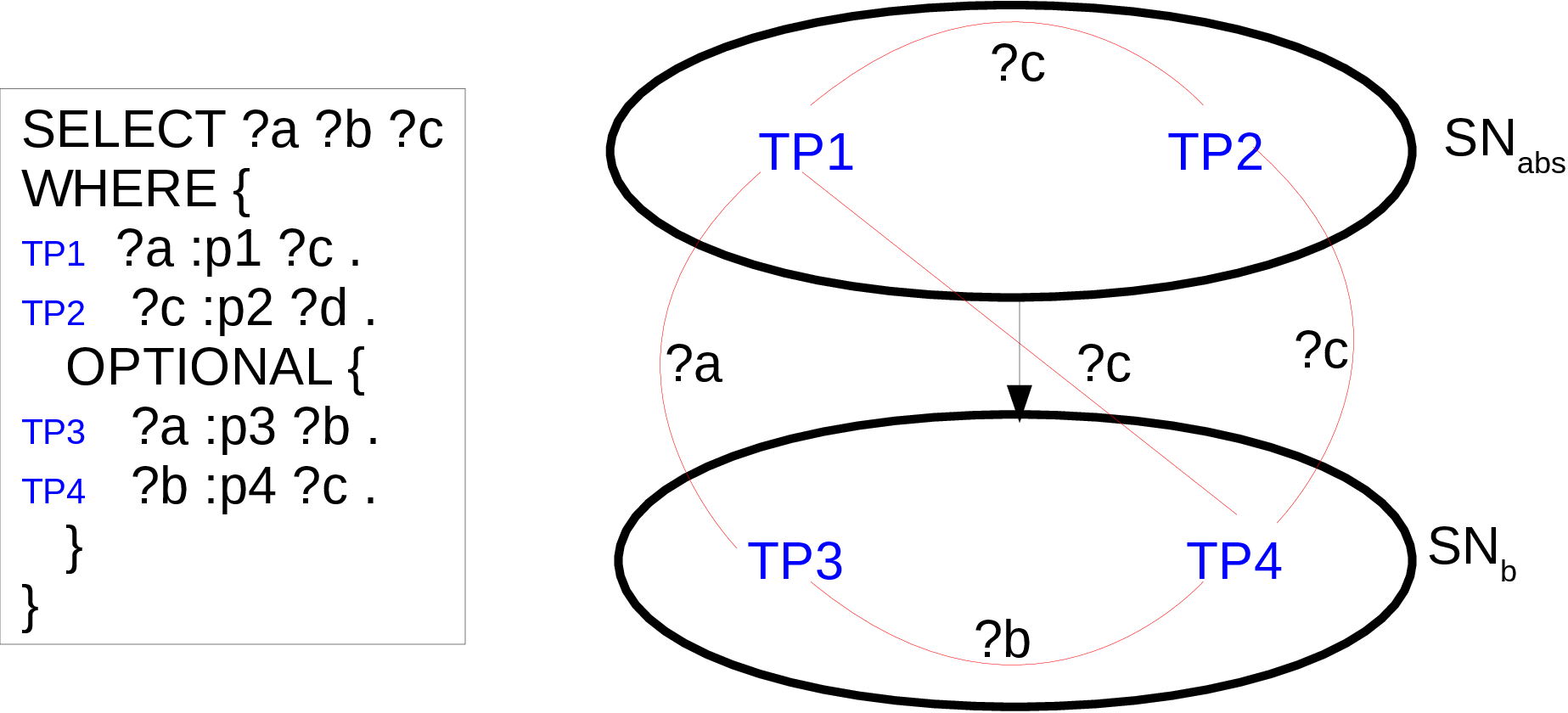}
 \caption{Example of a query where nullification and best-match cannot be 
avoided although GoTs of individual supernodes are acyclic} 
\label{fig:exception2}
\end{figure}

In Figure \ref{fig:exception2}, we have given an example of a query where there 
are two equivalence classes of GoT edges running between $SN_{abs}$ and the 
slave supernode $SN_b$, and hence it cannot avoid nullification and best-match. 
Notice carefully that in this query although the individual GoTs of $SN_{abs}$ 
and the slave supernode are acyclic, the GoT over all the triple patterns is 
cyclic.

\begin{lemma} \label{lemma:cyclic}
For a BGP-OPT query, if (1) there is only one equivalence class of edges 
across any pair of master-slave supernodes, and (2) each subgraph of GoT 
representing the triple patterns in each slave supernode is acyclic, then 
nullification and best-match can be avoided if the triples associated with the 
triple patterns are pruned with {\normalfont\texttt{prune\_triples}}
(Algorithm \ref{alg:prune}), and results are produced using 
{\normalfont\texttt{multi-way-join}} (Algorithm \ref{alg:multiway}). This holds 
even if the GoT of the triple patterns in $SN_{abs}$ is cyclic.
\end{lemma}

\begin{proof}
In Lemma \ref{lemma:acyclic}, we saw that if each OPT-free BGP component of a 
query has triple patterns with minimal triples associated with it, 
nullification and best-match can be avoided. When $SN_{abs}$ has a cyclic GoT, 
it cannot be guaranteed that the triple patterns in it have minimal triples 
associated with them after \texttt{prune\_triples} method. Note that we prune 
the triples in $SN_{abs}$ before any other slave supernodes in the query, and 
then do not revisit it again in \texttt{prune\_triples} (lines  
\ref{ln:orderslaves} and \ref{ln:eachsn}). Thus when a slave supernode fetches 
the variable bindings from its master (lines 
\ref{ln:ifslave}-\ref{ln:endifslave} in Algorithm \ref{alg:prune}), the master's 
variable bindings have already been frozen (irrespective of whether minimal or 
not). During \texttt{multi-\\way-join}, we again visit triple patterns in their 
master-slave hierarchy, i.e., we start with the triple patterns in $SN_{abs}$, 
create variable bindings for all the triple patterns in it, and then move on to 
other slave supernodes, in their respective master-slave hierarchy.

If the GoT of just triple patterns in $SN_{abs}$ has cycles, then we 
\textit{backtrack} while generating variable bindings for its triple patterns 
whenever there is a mismatch (ln \ref{ln:rollback} in Algorithm 
\ref{alg:multiway}). Thus when we start processing the slave triple patterns 
from an \textit{acyclic} slave supernode, all the variable bindings of 
$SN_{abs}$ have been decided, and are consistent across all the triple patterns 
of $SN_{abs}$. Having only one equivalence class of edges between $SN_{abs}$ 
and its slave, say $SN_{i}$, means that there is always \textit{one} triple 
pattern, say $T_m$, in $SN_{abs}$ and respectively $T_s$ in $SN_i$, which have 
one or more shared variables between them that cover all the shared variables 
between $SN_{abs}$ and $SN_{i}$. Then in \texttt{multi-way-join}, whenever we 
start processing the triple patterns from $SN_i$ for generating variable 
bindings, we first ensure that $T_s$ in $SN_i$ has the exact same bindings for 
\textit{all} the variables shared between it and $T_m$. It may happen that such 
a (composite) binding does not exist in $T_s$, then we immediately set all the 
variables in $SN_i$ to NULL bindings, and do the same for all of slaves of 
$SN_i$ too. Thus at the end of one iteration of \texttt{multi-way-join}, we do 
not have any inconsistent variable bindings in \texttt{vmap} that necessitate 
nullification, followed by best-match. We assume that the BGP-OPT queries are 
\textit{well-designed}, which ensures that no master-slave supernode that is 
not connected with a directed edge share variables among them, that are not 
shared between any intermediate master supernode.

Note that this is possible only because each slave supernode has an 
\textit{acyclic} GoT, and the triple patterns in it have minimal triples, that 
are consistent across the same slave supernode. Since there is only one 
equivalence class edge between any master-slave pair of supernodes, a unique 
triple pattern in the slave can determine if the respective bindings exist in 
the slave. \qed
\end{proof}

\section{UNIONs, FILTERs, DISTINCT} \label{sec:discuss}

Basic Graph Patterns (BGPs) \textit{connect} the triple patterns through 
inner-joins, and the OPTIONAL patterns connect two BGPs through 
left-outer-joins. BGP-OPT patterns make the basic building blocks of the 
SPARQL query language, as we will see in this section later through the query 
rewriting rules in UNIONs and FILTERs. SPARQL, like SQL, has other 
\textit{non-join} constructs 
too such as UNIONs, FILTERs, DISTINCT. SPARQL also has additional constructs 
like ORDER BY, FROM, FROM NAMED, REDUCED, OFFSET, LIMIT, ASK, CONSTRUCT etc. 
SPARQL grammar is recursive, and it can be pictorially represented as shown in 
Figure \ref{fig:sparqlgram}\footnote{For the clarity and conciseness, in 
Figure \ref{fig:sparqlgram} we have shown the core recursive performance 
intensive components of SPARQL. SPARQL 1.1 grammar has other syntactic 
components that are not shown in the figure.}. \fxnote{Show fig from thesis.}
In our work we have mainly focused on the SPARQL components that constitute the 
\textit{recursive} part of the language, and thus make the performance 
intensive components of the query evaluation strategies.
Having done the BGP-OPT component analysis in the previous sections, further 
in this section, we analyze UNIONs, FILTERs, and DISTINCT clauses, and their 
interaction with the BGP-OPT component.

\begin{figure}
 \centering
 \includegraphics[scale=0.3]{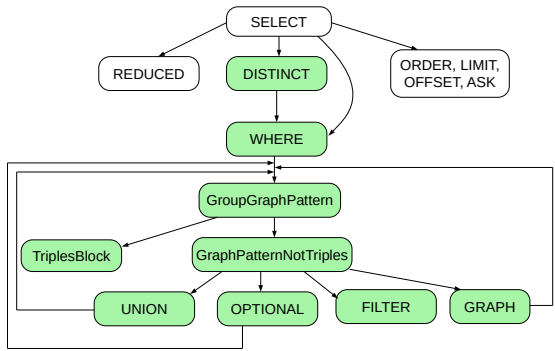}
 \caption{SPARQL grammar pictorial representation} \label{fig:sparqlgram}
\end{figure}

\subsection{UNIONs} \label{sec:union}
\fxnote{Give an example of UNION.}
The SPARQL grammar does \textit{not} enforce the two patterns being UNIONed to 
be \textit{union compatible} (SQL does)\footnote{Two patterns are 
union compatible, if they have the same \textit{arity} and same 
\textit{attributes}, e.g., for SPARQL it means that the two patterns would have 
the same \textit{size} and the same \textit{variables} in them.}. However, for 
our discussion, we assume \textit{well-designed UNIONs} -- for every subpattern 
$P' = (P_1 \cup P_2)$ in a query, if a variable $?j$ appears outside $P'$, then 
it \textit{must} appear both in $P_1$ and $P_2$. This assumption avoids having 
to deal with a query with ``dangerous'' variables, Cartesian products, and 
\textit{disconnected} GoT  \cite{polleres}.

For any BGP-OPT-UNION query, first we identify UNION-free BGP-OPT subcomponents 
of the query. E.g., in a query of the form $P_1 \cup ((P_2 \Join P_3) \cup 
(P_4 \leftouterjoin P_5))$, where $P_{1...5}$ are all OPT-free BGPs, and 
following are the UNION-free components -- (a) $P_1$, (b) 
$(P_2 \Join P_3)$, (c) $(P_4 \leftouterjoin P_5))$. For each such UNION-free 
BGP-OPT sub-component of the query, we prune the triples associated with each 
triple pattern in that component using our \texttt{prune\_triples} procedure 
(Algorithm \ref{alg:prune}). Note that for this pruning procedure, we 
\textit{ignore} all the other sub-components of the query and the respective 
triple patterns in those sub-components. We only consider the given 
sub-component as an independent BGP-OPT query.

Next we apply following three conversion rules on the BGP-OPT-UNION query and 
convert the query in the \textit{UNION Normal Form} (UNF) \cite{perez2}. Note 
that for the 
notations given below a pattern $P_i$ need not be a BGP, it can be a pattern  
with other sub-patterns nested inside it, and the $\equiv$ symbol indicates 
that the results generated by the queries on the either side of $\equiv$ are 
always the same for a \textit{well-designed union compatible} query.
\begin{enumerate}
 \item $(P_1 \cup P_2) \leftouterjoin P_3 
\equiv (P_1 \leftouterjoin P_3) \cup (P_2 \leftouterjoin P_3)$.
 \item $(P_1 \cup P_2) \Join P_3 \equiv (P_1 \Join P_3) \cup (P_2 
\Join P_3)$
 \item $P_1 \leftouterjoin (P_2 \cup P_3)$ is rewritten as 
$(P_1 \leftouterjoin P_2) \cup (P_1 \leftouterjoin P_3)$. However, $P_1 
\leftouterjoin (P_2 \cup P_3) \not\equiv (P_1 \leftouterjoin P_2) \cup (P_1 
\leftouterjoin P_3)$ for conventional union\footnote{However, if we do a 
\textit{minimum-union} the results will be the same as elaborated later.}.
We will elaborate on this rewrite after presenting our query evaluation 
technique below.
\end{enumerate}

We convert any BGP-OPT-UNION query in the UNF by applying the above three 
conversion rules, such that the resulting query looks like $P_i \cup P_{i+1}... 
\cup P_{i+k}$ where each $P_{i+j}, 0 \leqslant j \leqslant k$ is a UNION-free 
BGP-OPT pattern. Note that we have deliberately used suffixes $(i+j), 0 
\leqslant j \leqslant k$ for this representation, to avoid confusion with the 
patterns in a BGP-OPT-UNION query \textit{before} bringing it in the UNF. Now 
each of these 
subqueries can be treated independently as BGP-OPT 
queries for the purpose of entire query evaluation. For each such subquery, we 
run \texttt{multi-way-join} as given in Algorithm \ref{alg:multiway}. 
\fxnote{Modify algo to take care of these cond.} Note that 
\texttt{multi-way-join} needs to know if it has to do \texttt{nullification} on 
each generated result. Just like we do for any UNION-free BGP-OPT query, for 
each $P_{i+j}$ subquery in the UNF, nullification is required if it 
\textit{violates any} of the following conditions:
\begin{enumerate}[label=\alph*.]
  \item Considering the GoSN of $P_{i+j}$, GoT of triple patterns enclosed in 
each slave of GoSN is acyclic.
 \item For every pair of master-slave supernodes such as $SN_a \rightarrow 
SN_b$, we consider \textit{equivalence classes} of GoT edges $\langle T_m, 
T_s\rangle, T_m \in SN_a, T_s \in SN_b$ (recall our definition of equivalent 
classes of GoT edges from Section \ref{sec:prune}). Then for every pair of 
master-slave supernodes, there is only one equivalence class of GoT edges 
between master-slave triple patterns.
 \end{enumerate}

Once the final results are generated for each $P_{i+j}$ subquery in the UNF 
using the \texttt{multi-way-join} and nullification (wherever required), we do a 
``union all'' of results of all the subqueries -- all the results are compiled 
together along with any duplicates as well. Next, we decide when we need to use 
the best-match operation over the unioned results, and after that we elaborate 
\textit{why} we do this. Best-match is required if --

\begin{itemize}[label=$\bullet$]
 \item nullification operation was used in at least one subquery in the UNF.
 \item while bringing the original query in UNF, at least once pattern $P_1 
\leftouterjoin (P_2 \cup P_3)$ was rewritten as $(P_1 \leftouterjoin P_2) \cup 
(P_1 \leftouterjoin P_3)$.
\end{itemize}

For the first condition -- nullification used in any subqueries -- it is 
straightforward to see that best-match is required to remove the subsumed 
results (ref. Section \ref{sec:nullbm}). For the second condition, recall that 
for the third union expansion rule, we noted that $P_1 \leftouterjoin (P_2 \cup 
P_3) \not\equiv (P_1 \Join P_3) \cup (P_2 \Join P_3)$, and now we elaborate on 
this. Per SPARQL grammar, the UNION operation does a ``union all'' of the 
results produced by the unioned patterns, without removing the duplicates. Also 
when the unioned elements have $>1$ arity, the \textit{set union} operation does 
not remove \textit{subsumed} results (ref. Section \ref{sec:nullbm} for the 
definition of \textit{subsumed} results.).

\begin{definition}
 A union operation that removes subsumed results is called \textbf{minimum 
union}.
\end{definition}

Let two sets of bindings for variable pairs $(?a, ?b)$ be unioned as  
\textit{\{(:p1, NULL)\}} $\cup$ \textit{\{(:p1, :p2)\}}. The result of this 
union is \textit{\{(:p1, NULL), (:p1, :p2)\}}. However, if the sets of 
variable bindings being unioned are \{(NULL, NULL)\} $\cup$ 
\textit{\{(:p1, :p2)\}}. Then the result is \textit{\{(:p1, :p2)\}}. This means 
that union of \textit{NULL} co-existing with another non-null value stays in 
the final results despite being subsumed by another result, but not by itself.
Also in SPARQL, unlike most SQL systems, the 
joins are \textit{null-compatible}, i.e., the join of variable bindings of 
$(?a, ?b)$, $\{(:p1, NULL)\} \Join (:p1, :p2)$ is $\{(:p1, :p2)\}$, and 
$\{(NULL, NULL)\} \Join (:p1, :p2) = \{(:p1, :p2)\}$. Because of the above two 
important observations, if we rewrite a pattern $P_1 \leftouterjoin (P_2 \cup 
P_3)$ as $(P_1 \leftouterjoin P_2) \cup (P_1 \leftouterjoin P_3)$, it may 
generate different results than $P_1 \leftouterjoin (P_2 \cup P_3)$, if either 
$P_2$ or $P_3$ or both generate all null results.

Thus, for $P_1 \leftouterjoin (P_2 \cup P_3)$ pattern rewritten as $(P_1 
\leftouterjoin P_2) \cup (P_1 \leftouterjoin P_3)$, the final results 
may not match. However the original query and its UNF rewrite are 
semantically the same, and a query execution method that employs minimum-union 
instead of union-all or set-union generates all the \textit{non-duplicate} and 
\textit{non-subsumed} results correctly.

Thus we conclude this discussion that \textit{minimum union} instead of 
set-union or union-all, allows the third rewrite rule for converting a 
BGP-OPT-UNION query into UNF while keeping the patterns before and after the 
rewrite equivalent with respect to the generated results. As a result, it 
allows any BGP-OPT-UNION query to be brought in the UNF allowing us more 
options for query optimization strategies such as the ones described as a part 
of this article.

In our previous work \cite{atresigmod15}, we had proposed to convert a
BGP-OPT-UNION query into the UNF and process each of the subqueries in the UNF 
independently. This involves duplication of efforts to prune the triples 
associated with a triple pattern that appears in multiple UNF subqueries (due 
to rewrite rules). Through the method presented in this section, we avoid this 
by treating each union-free BGP-OPT component of the query independently, 
pruning it \textit{before} bringing the query in UNF, and use normal sequence 
of \texttt{multi-way-join}, followed by \textit{best-match} wherever required 
by the structure of the query.

\begin{lemma}
Nullification and best-match can be avoided if each subquery $P_{i+j}, 1 
\leqslant j \leqslant k$ in the Union Normal Form of a BGP-OPT-UNION query 
satisfies the following 
two conditions:
\begin{enumerate}[label=\alph*.]
  \item Considering only the GoSN of $P_{i+j}$, GoT of triple patterns 
enclosed in each slave supernode of this GoSN is acyclic.
 \item For every pair of master-slave supernodes such as $SN_a \rightarrow 
SN_b$ in the GoSN of $P_{i+j}$, there is only one equivalent class GoT edges.
 \end{enumerate}
\end{lemma}

\begin{proof}
This lemma is obtained by combining the results of Lemma \ref{lemma:acyclic} 
and \ref{lemma:cyclic}, and hence the proof follows from the proof of those 
lemmas. This is because each $P_{i+j}$ subquery in UNF can be treated as a 
BGP-OPT query independently for the purpose of generating the final results of 
the query. \qed
\end{proof}

\subsection{FILTERs} \label{sec:filter}
SPARQL FILTER construct can have a complex Boolean expression that is supposed 
to be evaluated for every answer/result generated for the pattern over which 
it is applied. Below we have given an example of a BGP-OPT-FILTER query.

\vspace{2mm}
\begin{lgrind}
\Head{}
{\small
\L{\LB{\K{SELECT}_?\V{friend}_?\V{sitcom}_?\V{dir}}}
\L{\LB{}\Tab{4}{\K{WHERE}_\{}}
\L{\LB{}\Tab{6}{:\V{Jerry}_:\V{hasFriend}_?\V{friend}_.}}
\L{\LB{}\Tab{6}{?\V{friend}_:\V{age}_?\V{age}_.}}
\L{\LB{}\Tab{6}{?\V{friend}_:\V{actedIn}_?\V{sitcom}_.}}
\L{\LB{}\Tab{6}{\K{OPTIONAL}_\{}}
\L{\LB{}\Tab{8}{?\V{sitcom}_:\V{hasDirector}_?\V{dir}_.}}
\L{\LB{}\Tab{8}{?\V{sitcom}_:\V{location}_:\V{NYC}_.}}
\L{\LB{}\Tab{6}{\}}}
\L{\LB{}\Tab{6}{\K{FILTER}(?\V{age} \< \V{60} \&\& ?\V{dir} != \V{:Jerry})}}
\L{\LB{}\Tab{4}{\}}}
}
\end{lgrind}
\vspace{2mm}

In essence, this query is asking for all the friends of \textit{:Jerry} who 
have acted in a \textit{?sitcom}, and optionally information of the 
\textit{?dir} of the \textit{?sitcom} too, if it was located in \textit{:NYC}. 
Notice that the FILTER condition further emphasizes that the \textit{?friend} 
must be younger than ``60'' years of age, and the \textit{?sitcom}'s director 
(\textit{?dir}) \textit{cannot} be \textit{:Jerry} himself.

In general, FILTER expression of a SPARQL query can be notationally 
represented as $P_x \mathcal{F}(R)$, where $P_x$ is a SPARQL query pattern 
-- BGP, OPTIONAL, UNION or any combination of these -- and $R$ is a Boolean 
valued filter condition to be applied on $P_x$. For the scope of our discussion, 
we assume \textit{safe filters} \cite{perez2,polleres}, i.e., for $P_x 
\mathcal{F}(R)$, all the variables in $R$ appear in $P_x$, $vars(R) \subseteq 
vars(P_x)$. This is in line with our previous discussion of well-designed 
OPTIONAL and UNION patterns. \textit{Unsafe} (non-well-designed) filters can 
alter the semantics of the OPTIONAL patterns (ref \cite{perez2} for more 
detailed discussion). The FILTER pattern evaluation techniques presented here 
can work with any combination of BGP, OPTIONAL, and FILTER, because BGP-OPT 
remain as the building blocks of a SPARQL query, and our query evaluation 
techniques are built for these basic blocks.

\fxnote{Give an example of filter.}
For a SPARQL query with FILTERs, first we \textit{push-in} the filter 
conditions as much as possible using the following rewrite rule on each 
UNION-free subcomponent of the query \textit{without} bringing the query in the 
UNF (the first three rewrite rules are described in Section \ref{sec:union}).

\begin{enumerate}
\setcounter{enumi}{3}
 \item $(P_1 \leftouterjoin P_2)\mathcal{F}(R) \equiv (P_1 \mathcal{F}(R)) 
\leftouterjoin P_2$.
\end{enumerate}

This rule can be applied as our queries are assumed to be 
\textit{well-designed}. After pushing-in the filters, we run 
\texttt{prune\_triples} on each UNION-free subcomponent of the query same as 
described in Section \ref{sec:union}. We have an option to apply the FILTER 
conditions while loading the BitMat associated with each triple pattern. 
However, we can do so only if the respective FILTER expression satisfies 
following constraints.
\begin{enumerate}
 \item FILTER expression is composed of conjunctive expression 
(only ``$\&\&$'' and no ``$||$''), and
 \item each subexpression in the conjunction 
consists of only one variable (e.g., \textit{FILTER(?a $>$ 60 $\&\&$ ?b $!=$ 
10)}).
\end{enumerate}

Then each subexpression of FILTER can be applied while loading the BitMat 
associated with the triple pattern that contains the respective variable in the 
FILTER subexpression. However if a FILTER expression consists of disjunction 
(``$||$''), it can only be applied after \texttt{multi-way-join}, once each 
result of the query is fully constructed. The decision of applying a filter 
during BitMat loading will also depend on the \textit{selectivity} of the filter 
subexpression, and the available indexes. If these constraints are not 
satisfied, FILTER can be applied only in \texttt{multi-way-join}. For 
simplicity, and the scope of this article, we assume that all the FILTER 
expressions in a BGP-OPT-UNION-FILTER query are applied only in the 
\texttt{multi-way-join} procedure, as a part of the \texttt{nullification} 
process.

Note from Section \ref{sec:union}, that if a query contains UNI-ONed patterns, 
we first prune the triples in the UNION-free subcomponents of the query using 
\texttt{prune\_triples} (Algorithm \ref{alg:prune}), then bring the query in the 
UNF, and then evaluate each subquery independently using 
\texttt{multi-\\way-join}. For a query with UNIONs and FILTER conditions, after 
\texttt{prune\_triples}, we bring the query in UNF by applying rewrite rules 
1--3 described in Section \ref{sec:union}, rule 4 described earlier in this 
section, and additionally applying rule 5 below.
\begin{enumerate}
\setcounter{enumi}{4}
 \item $(P_1 \cup P2) \mathcal{F}(R)$ $= (P_1\mathcal{F}(R)) \cup 
(P_2\mathcal{F}(R))$.
\end{enumerate}

Any FILTER expressions that were not applied during BitMat loading are applied 
as a part of the \texttt{nullifi-\\cation} procedure, when 
\texttt{multi-way-join} produces each result. \texttt{Nullification} not only 
ensures consistent variable bindings (for any reordered inner and left-outer 
joins), but also evaluates the filter conditions. If the filter condition fails 
for the given pattern, we nullify all the variable bindings of that pattern and 
any slaves of the pattern. Recall that here the GoSN of the query comes in handy 
for quickly determining the nullification of any slave patterns. Once the 
results of all the subqueries on UNF are generated, we union them all, and apply 
\texttt{best-match} (\textit{minimum union}) \textit{iff} $(P_1 \leftouterjoin 
(P_2 \cup P_3))$ was rewritten as $(P_1 \leftouterjoin P_2) \cup (P_1 
\leftouterjoin P_3)$ for some subquery, or at least one variable binding was 
nullified during the \texttt{nullification} operation in 
\texttt{multi-way-join}, either as a result of a cyclic subquery or the filter 
condition.

\subsection{DISTINCT} \label{sec:distinct}
A DISTINCT keyword in the SELECT clause eliminates duplicates from the results. 
Consider the following BGP query over an RDFized version of a movie database 
like IMDB, which is asking for all the \textit{distinct} pairs of the actors 
($?a$) and their directors ($?d$).

\vspace{2mm}
\begin{lgrind}
\centering
{\small
\L{\LB{\K{SELECT}_\K{DISTINCT}_?\V{a}_?\V{d}}}
\L{\LB{\K{WHERE}\{}}
\L{\LB{}\Tab{4}{?\V{m}_\V{rdf}:\V{type}_:\V{Movie}_.}}
\L{\LB{}\Tab{4}{?\V{m}_:\V{hasActor}_?\V{a}_.}}
\L{\LB{}\Tab{4}{?\V{m}_:\V{hasDirector}_?\V{d}_.\}}}
}
\end{lgrind}
\vspace{2mm}

\textit{:UmaThurman} has acted in three movies directed by 
\textit{:QuentinTarantino}, they are \textit{:PulpFiction}, 
\textit{:KillBillVol1}, and \textit{:KillBillVol2}. Without the DISTINCT 
clause, this query will generate three copies of (\textit{:UmaThurman}, 
\textit{:QuentinTarantino}) as the variable bindings for \textit{(?a, ?d)} in 
the results. With the DISTINCT clause we will get only one copy.

SPARQL algebra allows an arbitrary number of variables in the DISTINCT clause 
(just like SQL). When there are multiple variables in the DISTINCT clause 
(like in the example above), the distinct values are \textit{composite} of the 
bindings of the respective variables, e.g., (\textit{:UmaThurman}, 
\textit{:QuentinTarantino}) is distinct from \\(\textit{:UmaThurman}, 
\textit{:WoodyAllen}), although they both share \textit{:UmaThurman}. If the 
variables in the DISTINCT clause appear in different triple patterns, like 
in our example, we have to generate intermediate variable bindings of the other 
variables not in the DISTINCT clause, and discard them later. This may create a 
memory overhead. E.g., we first have to generate bindings of all three 
variables \textit{(?m, ?a, ?d)}, project out only bindings of \textit{(?a, ?d)}, 
and then pass these \textit{(?a, ?d)} binding pairs through the DISTINCT filter 
to remove duplicates. Hence for an arbitrary number of variables in the WHERE 
and DISTINCT clauses, evaluation of a query may become memory intensive if the 
DISTINCT clause is composed of many variables that do not appear in the same 
triple pattern.

For a SPARQL query with any intermix of BGP, OPTIONAL, UNION, FILTER, we employ 
different techniques for the evaluation of the DISTINCT clause depending on the 
other clauses in the query. Since BGP-OPT patterns are the basic building 
blocks of a SPARQL query, we begin with how to handle the DISTINCT clause with 
BGP and OPTIONAL patterns so as to avoid using extra memory overhead for the 
removal of non-essential variables, e.g., \textit{?m} in our example above.

\subsubsection{Acyclic BGP} \label{sec:distacybgp}
For a SPARQL query with just a basic graph pattern, i.e., 
inner-joins alone, if it is \textit{acyclic} per the definition of acyclicity 
as presented in Section \ref{sec:acyclic}, we prune the triples using 
\texttt{prune\_triple} (Algorithm \ref{alg:prune}), which leaves 
\textit{minimal} triples in each BitMat associated with the triple 
patterns. \texttt{Mult-way-join} projects out all the variable bindings without 
duplicate removal. If we project out only some variables, we may get duplicates 
(in composite value form), as seen in the example elaborated at the beginning of 
this section. Hence we need a way to remove the bindings of 
\textit{non-essential} variables, i.e., those not appearing in the DISTINCT 
clause, and remove duplicates from the composite bindings of the DISTINCT 
variables. For this discussion we assume the query to not have OPTIONAL, UNION, 
or FILTER patterns. Thus the GoSN of the query is going to have only one 
supernode, $SN_{abs}$, and all the triple patterns are going to be encapsulated 
inside it. Considering the \textit{graph of triple patterns} (GoT) with 
undirected edges between the triple patterns, we identify a \textit{Minimal 
Covering Subgraph} (MCS) such that triple patterns in MCS cover all the 
variables appearing in the DISTINCT clause, and no other strict subgraph of 
this MCS covers DISTINCT variables -- in short we identify the minimum triple 
patterns from GoT that are required to project out the DISTINCT variable 
bindings. \fxnote{How to identify MCS}

\begin{figure}[h]
 \centering
 \includegraphics[scale=0.4]{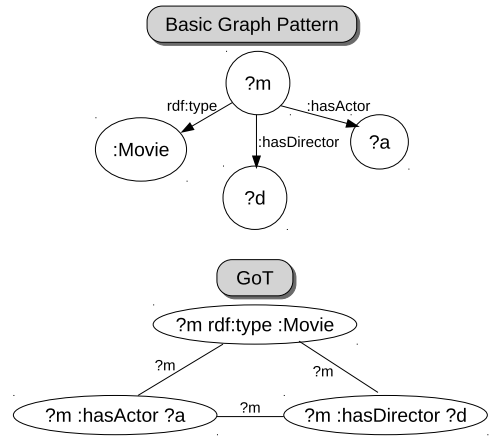}
 \caption{BGP and GoT of the example query}\label{fig:distinctq}
\end{figure}

MCS can be identified methodically as follows. 
First identify all the triple patterns such that it has one or more DISTINCT 
variables in it, and then identify a minimal subgraph that connects all 
these triple patterns. Eliminate a $T_i$ from this subgraph, if it has at least 
one neighbor $T_j$ in this subgraph such that the $dist\_vars(T_i) \subseteq 
dist\_vars(T_j)$, i.e., the DISTINCT variables appearing in $T_i$ also appear 
in $T_j$, and elimination of $T_i$ does not disconnect the rest of the 
subgraph. We continue this process until we do not find any triple pattern to 
eliminate. This makes the \textit{minimal covering subgraph} (MCS).
This MCS is \textit{acyclic}, because the original GoT from which the 
MCS is carved out is acyclic too. Thus conceptually this MCS represents a 
\textit{subquery} of the original BGP query. Now we will operate on this MCS to 
eliminate \textit{non-essential} variables -- those \textit{not} appearing in 
the DISTINCT clause but appearing in the MCS because they connect one or more 
DISTINCT variables. Before that we review some important properties of 
\textit{Boolean matrix multiplication} (BMM) of BitMats. BitMats conceptually 
represent an RDF graph's adjacency matrices.

Figure \ref{fig:distinctq} shows the Basic Graph Pattern (BGP) and the graph of 
triple patterns (GoT) of the query given at the beginning of this section. 
Triple patterns \textit{?m :hasActor ?a} and \textit{?m :hasDirector ?d} are 
connected to each other with 
label \textit{?m} in the GoT. Note that these two triple patterns are part of 
the MCS due to \textit{?a} and \textit{?d} variables, but the triple pattern 
\textit{?m rdf:type :Movie} is not. The BitMats associated with these triple 
patterns contain all the triples which have predicates (edge-labels) 
\textit{:hasActor} and \textit{:hasDirector} respectively, and they in turn 
represent the adjacency matrices of two subgraphs of the original RDF graph,  
with only \textit{:hasActor} and \textit{:hasDirector} edge labels respectively. 
The transpose of the BitMat of \textit{:hasDirector} conceptually reverses 
the direction of edges between the respective RDF nodes and represents a 
triple pattern \textit{?d :hasDirector ?m}. Thus if we do a BMM of the 
\textit{transpose} of BitMat of \textit{:hasDirector} with the BitMat of 
\textit{:hasActor}, the resultant matrix represents all the \textit{distinct} 
pairs of nodes that have \textit{at least one undirected path} with edge labels 
\textit{:hasDirector--:hasActor} between them, and eliminates the bindings of 
\textit{?m}. It, in turn, eliminates the RDF nodes representing 
\textit{?m}, which is a common variable between the two triple patterns. 
Figure \ref{fig:distinctq2} pictorially represents this concept of Boolean 
Matrix Multiplication. Note that since this query is \textit{acyclic}, and we 
have already done \texttt{prune\_triple}, $BitMat_1$ and $BitMat_2$ have 
\textit{minimal} triples left in it, thus their BMM gives us the exact distinct 
pairs of $(?a, ?b)$ in $BitMat_{12}$.

\begin{figure}[h]
 \centering
 \includegraphics[scale=0.38]{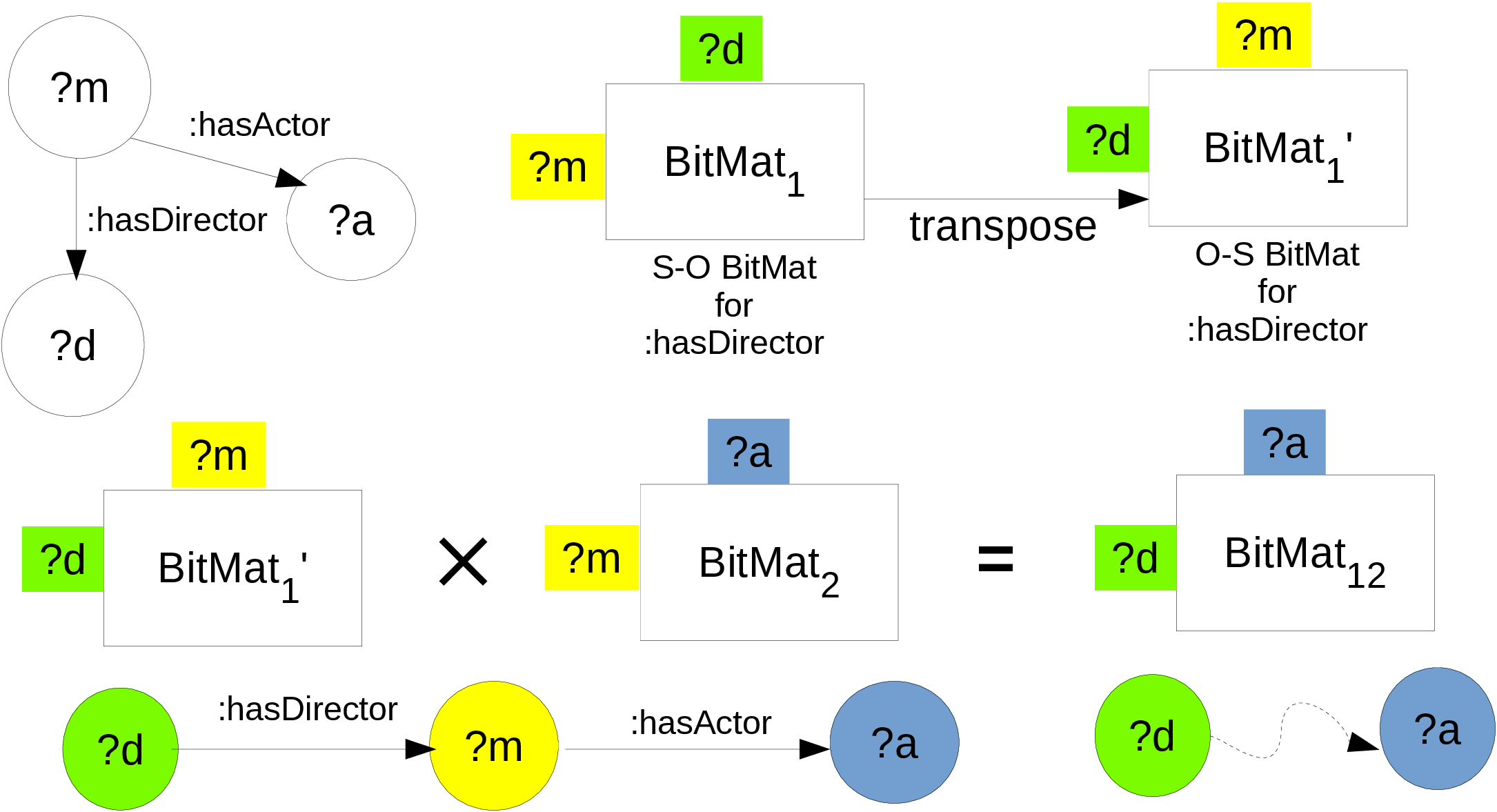}
 \caption{Boolean Matrix Multiplication}\label{fig:distinctq2}
\end{figure}

\begin{figure*}[t]
 \centering
     \includegraphics[scale=0.55]{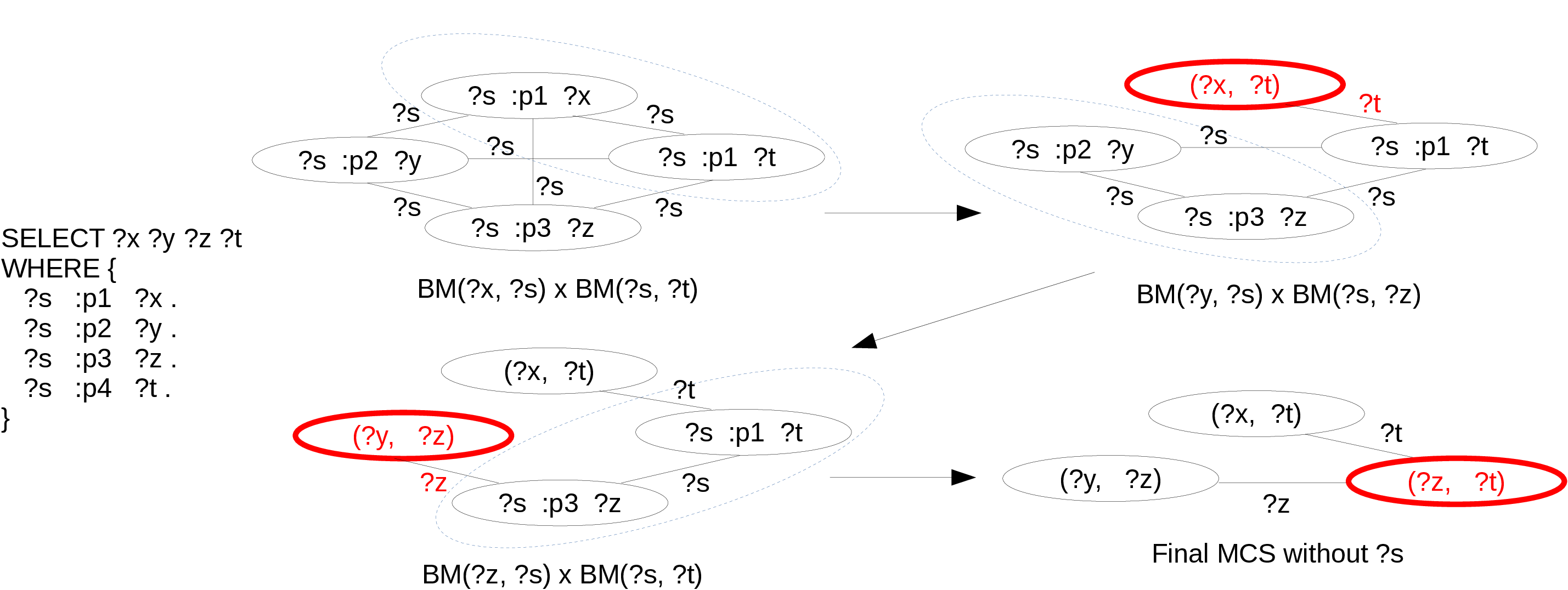}
     \caption{Remove ?s from MCS using the algorithm}\label{fig:elimy}
 \end{figure*}

\begin{prop}
 An edge between two triple patterns in a minimal covering subgraph signifies a 
potential Boolean matrix multiplication between them to eliminate the variables 
common between them, which appear as the edge label.
\end{prop}

We make use of this property to methodically eliminate non-essential variables 
in the MCS to shrink it further, and get a set of BitMats absolutely required 
to project the DISTINCT variable bindings. The MCS shrinking process is carried 
out as follows.\fxnote{Draw alt fig with GoT.}

\begin{enumerate}
 \item For every pair of triple patterns $T_i$ and $T_j$ that have an 
undirected edge between them with label, say \textit{?m}, such that $?m$ is a 
non-essential variable, do a BMM of $T_i$ and $T_j$ such that bindings for ?m 
get eliminated. Whether we need to take a transpose of the BitMat or not 
depends on whether we have loaded S-O or O-S BitMat for the respective triple 
pattern (ref Section \ref{sec:index}). E.g., in our example if we have loaded 
O-S BitMat \textit{(?m :hasDirector ?d)} and S-O BitMat of \textit{(?m 
:hasActor ?a)}, then we do a BMM of $BM_{:hasDirector} \times 
BM_{:hasActor}$ with no change in the respective BitMats. However if we have 
loaded S-O BitMat of \textit{(?m :hasDirector ?d)} or O-S BitMat of \textit{(?m 
:hasActor ?a)}, we take the transpose of them before doing the BMM. Let us 
denote the resultant BitMat as $T_{ij}$.

\item $T_{ij}$ can connect to any neighbor of $T_i$ or $T_j$ (excluding $T_i$ 
and $T_j$), if it shares the exact same variables with the respective neighbor 
as done by $T_i$ or $T_j$. If $T_{ij}$ can thus connect to all of $T_i$ and 
$T_j$'s neighbors, we remove both $T_i$ and $T_j$ by connecting $T_{ij}$ to 
their neighbors.

\item $T_{ij}$ may not be able to connect to all of the $T_i$ and $T_j$'s 
neighbors, if $T_i$ or $T_j$ connects with other triple patterns over the 
same edge label between $T_i$ and $T_j$, e.g., $?m$. In that case we keep 
either of $T_i$ or $T_j$, eliminate the other. We connect $T_{ij}$ to the 
preserved triple pattern and eliminated triple pattern's neighbors (whichever 
it can connect to). This procedure of preserving either of $T_i$ or $T_j$ is 
explained further. If $T_i$ is connected to $T_j$ over, say 
$?s$, $T_{ij}$ cannot connect to any neighbors of $T_i$ and $T_j$ that are 
connected over $?s$. By preserving either of $T_i$ or $T_j$, we ensure that 
the remaining MCS remains connected.

\item We continue this process of eliminating triple patterns and their 
respective BitMats, until we have an MCS where each edge label contains a 
DISTINCT variable, or it has only one BitMat.
\end{enumerate}

In Figure \ref{fig:elimy}, we have shown a sample query where we need to 
eliminate the bindings of $?s$, while preserving correlations between the 
bindings of $?x$, $?y$, $?z$, $?t$. The figure also shows an evolution of this 
MCS to eliminate $?s$ using the above algorithm. Intuitively, we eliminate all 
the intermediate \textit{non-required} variables, by establishing a direct 
correlations between the bindings of the required variables. E.g., when $?x$ and 
$?t$ are part of two different triple patterns connected over $?s$, bindings of 
$?x$ and $?t$ are correlated through $?s$. When we do a BMM, $BM(?x, ?s) \times 
BM(?s, ?t)$ = $BM(?x, ?t)$, we establish a direct correlation between the 
bindings of the $(?x, ?t)$ pair, and eliminate the need of having $?s$ as an 
intermediary.

This algorithm is \textit{monotonic} -- at the end of one iteration, the edges, 
nodes, and BitMats in an MCS remain the same or become fewer than before. It 
gradually eliminates the edges with join variables that do not appear in the 
DISTINCT clause. At the end of the procedure, we are left with an MCS with 
BitMats -- either original or new ones created in the process of BMM -- and 
edges with only DISTINCT variables. This algorithm always converges when 
all the non-required variables are eliminated from the MCS. Also note that the 
total BitMats at the end of the algorithm are always \textit{fewer} than the 
original BitMats in the query -- note that in step 2, we \textit{remove two} 
BitMats while creating a \textit{new one}, and in step 3 we remove one BitMat 
and create one. Hence eventually we are left with \textit{fewer} BitMats -- thus 
reducing the memory requirements.

We join these BitMats with each other using the same \texttt{multi-way-join}
procedure (Algorithm \ref{alg:multiway}). Note that we can \textit{carve} out 
an MCS from the original GoT, because the query is \textit{acyclic}, and each 
triple pattern in the query has \textit{minimal} triples after 
\texttt{prune\_triples} (Algorithm-\ref{alg:prune}).

\subsubsection{Acyclic BGP-OPT queries} \label{sec:distbgpopt}
For the queries with an intermix of BGP and OPTIONAL patterns, we consider 
queries whose (a) GoT is completely acyclic, and (b) queries with cyclic GoT.

Like an acyclic BGP query, for an acyclic BGP-OPT query, 
\texttt{prune\_triples} ensures minimal triples associated with each triple 
pattern in the query. We eliminate non-essential variables before doing 
\texttt{multi-way-join} as follows. Starting with $SN_{abs}$, we go over all 
the supernodes in the master-slave hierarchy (masters before their respective 
slaves). We mark a supernode if it contains one or more variables appearing in 
the DISTINCT clause that do \textit{not} appear in any of its masters. This is 
because a variable's binding from the master triple pattern always dominates 
the binding from a slave triple pattern. So there is no need to consider a 
slave supernode if all of its DISTINCT variables are covered by one or more of 
its master. After this, we carve out a \textit{Minimal Covering GoSN} (MCGoSN) 
such that all the marked nodes are in a \textit{minimal connected subgraph} of 
GoSN.

An MCGoSN is connected if there there is an \textit{undirected path} from one 
supernode to another disregarding the directionality of the edges connecting 
the supernodes (recall our GoSN construction from Section \ref{sec:gosn}). 
Since we have a unique $SN_{abs}$ in our GoSN, any connected MCGoSN always 
contains $SN_{abs}$. Now there are two cases -- (1) $SN_{abs}$ and every 
supernode in MCGoSN contains at least one DISTINCT variable, (2) all the 
DISTINCT variables are contained in one or more slaves. For the second case, 
we cannot use the Boolean Matrix Multiplication technique to eliminate 
non-essential variables, and thus for this type of query we have to resort to 
standard way of listing all the bindings of DISTINCT variables, and then removing 
duplicates.

For the first case however, we can use the BMM technique as follows. We order 
all the supernodes in MCGoSN per master-slave hierarchy. We start with 
$SN_{abs}$ and the triple patterns in it, and carve out a \textit{Minimal 
Covering Subgraph} (MCS) from the GoT of only these triple patterns that cover 
the DISTINCT variables that appear in $SN_{abs}$. Next we iteratively go over 
the next supernode in the master-slave order, say $SN_i$. We extend the 
previously carved MCS to \textit{minimally} cover the triple patterns in $SN_i$ 
for the DISTINCT variables that appear \textit{only} in $SN_i$, but not in any 
of its masters. At the end of this exercise we get an MCS that contains triple 
patterns from different supernodes.

Next, we need to eliminate the non-essential join variables and triple patterns 
from this MCS in the same way we did for pure BGP queries. However, in this 
MCS, we may have a master-slave hierarchy between the two neighboring 
triple patterns. We take care of this hierarchy as follows. Pairs of 
connected triple patterns, $(T_i, T_j)$ can be categorized as -- (a) both 
$T_i$ and $T_j$ $\in$ $SN_{abs}$, (b) $T_i$ is a master of $T_j$ or vice versa, 
(c) $T_i$ and $T_j$ are slave peers contained in the same slave supernode, (d) 
$T_i$ and $T_j$ are contained in different slave supernodes.

Starting with type (a) pairs first, we completely reduce the part of MCS 
contained in $SN_{abs}$. This reduction will cause some changes in the MCS 
nodes and connections. Next we consider all (b) type pair of nodes in MCS. Let 
us assume $T_i$ is the master of $T_j$. If the edge-label between $(T_i, T_j)$ 
is a non-essential variable, we do a BMM, remove $T_j$, and reconnect $T_i$ to 
the newly created BitMat. Once we are done with all type (a) and (b) pairs, we 
shrink (c) type pairs same as the acyclic BGP technique. We ignore the (d) type 
edges. That is because recall that we assume OPTIONAL pattern queries to be 
\textit{well-designed}, so even if there are (d) type edges in the MCS, $T_i$ 
and $T_j$ always have an indirect path connecting them through their masters. 
Thus we gradually shrink this MCS to leave only the BitMats containing 
the bindings of the DISTINCT variables. Then we run \texttt{multi-way-join} on 
these BitMats in the standard way in the master-slave hierarchical order.

\subsubsection{Cyclic BGP-OPT queries}
For the cyclic BGP-OPT queries, the minimality of triples after 
\texttt{prune\_triples} cannot be guaranteed, so we cannot use this technique.  
For such queries, we have to enumerate the results with duplicates, and then 
remove of them using a na\"{\i}ve method.

\subsubsection{UNION, FILTER}
For queries with UNION or FILTER clauses along with BGP and OPT patterns, we 
cannot use the optimization technique of carving out an 
MCS from the entire query. The reason is -- we can carve out an MCS for 
DISTINCT variable binding projections only if each triple pattern in the query 
has \textit{minimal} triples associated with it. In a BGP-OPT query with UNION 
pattern, we do the pruning of UNION-free BGP-OPT sub-patterns in the given query 
(as described in Section \ref{sec:union}). Then get the query in the UNF $P_1 
\cup P_2 \cup ... \cup P_k$, and perform \texttt{multi-way-joins} on each $P_i$ 
in the UNF. However, since we did the pruning only on the UNION-free BGP-OPT 
subparts of the original query, the triples associated with the each triple 
pattern in $P_i$ may not be minimal.

The method presented in Section \ref{sec:filter} for handling FILTERs in 
DISTINCT-free queries shows that an arbitrary FILTER condition can cause 
\textit{nullification} while doing \texttt{multi-way-joins}, thereby altering 
the \textit{minimality} of triples associated with the triple patterns 
on-the-fly. Hence for the SPARQL queries asking for DISTINCT projection of some 
variables, where the query body has UNIONs or FILTER conditions, we use the 
na\"{\i}ve way of projecting out all the variable bindings of the DISTINCT 
variables, and then sorting to remove any duplicates.

\section{Related Work} \label{sec:relwork}
SPARQL BGP queries are similar to SQL inner-joins, and thus naturally a lot of 
SQL inner-join optimization techniques have been applied to SPARQL BGP query 
optimization. While the BGP query optimization has got a lot of attention, 
the discussion about optimization of other SPARQL components such as OPTIONAL 
patterns, UNIONs, FILTERs, DISTINCT clauses is quite sparse. This is despite 
the fact that in the context of SPARQL queries, these other components do make 
as large as 94\% of the queries \cite{usewod11,swim,manvmachine,practsparql}.
Previous work \cite{arenas,iswc14,perez2,schmidt}, has extensively 
analyzed the semantics of \textit{well-designed} OPTIONAL patterns from the 
perspective of \textit{tractability} properties. However, these texts have not 
focused on the discussion of other SPARQL components that are covered in this 
article. The idea of query graph of supernodes (GoSN) presented in this paper  
is reminiscent of \textit{well-designed pattern trees} (WDPT) 
\cite{perez,letelier2}, but WDPTs are undirected and unordered, whereas GoSN is 
directed, and establishes an order among the patterns (\textit{master-slave}, 
\textit{peers}), which is an integral part of our optimization techniques.
Also while previous discussion has focused on WDPTs and tractability results, 
they have not taken into consideration the aspects of optimization techniques 
from the point of view of \textit{minimality} of triples, and the \textit{order} 
of processing \textit{semi-joins} and \textit{multi-way-joins}, which make the 
performance intensive components of query evaluation techniques.

Galindo-Legaria, Rosenthal 
\cite{galindosigmod,galindo-legaria1,galindo-legaria2} and Rao et al 
\cite{rao2,rao1} have proposed ways of achieving SQL left-outer-join 
optimization through reordering inner and left-outer joins. Their work is 
closest to the work in this article. Rao et al have proposed 
\textit{nullification} and 
\textit{best-match} operators to handle inconsistent variable bindings and 
subsumed results respectively (see Section \ref{sec:nullbm}). In their 
technique, nullification and best-match are required for \textit{each} reordered 
query, as the minimality of tuples is not guaranteed. They do not use methods 
like \texttt{prune\_triples} to eliminate unwanted tuples before joins. 
Bernstein et al and Ullman \cite{semij2,semij1,ullman} have proved the 
properties of \textit{minimality} for \textit{acyclic inner-joins} only. 
Through our work, we have taken a major step forward by extending these 
properties in the context of SPARQL OPTIONAL patterns (SQL left-outer-joins), 
and finding ways to avoid overheads like \textit{nullification} and 
\textit{best-match} operations. We have also extended our previous results 
presented in \cite{atresigmod15} about the class of queries that can avoid 
nullification and best-match despite reordering of inner and left-outer joins.

Additionally, in this article we have extensively analyzed the UNION, FILTER, 
and DISTINCT clauses of SPARQL from the point of view of \textit{minimality} of 
triples, ways of unioning the results (\textit{minimum-union} versus standard 
union-all), and structural aspects of queries (cyclicity). With this analysis we 
have shown that a large number of  -- acyclic as well as some 
(good) cyclic -- SPARQL queries can be optimized by using our techniques of 
BGP-OPT query processing as \textit{building blocks}, and can avoid the 
additional overheads of processing and indexing for the correctness of the 
results. Our article also throws a new light on the treatment of \textit{NULL} 
values in the RDF and SPARQL context, and the implications of various SPARQL 
components -- other than Basic Graph Patterns (inner-joins) -- in the presence 
of NULL values in the query results. To the best of our knowledge, this topic, 
which directly affects implementation and optimization methods, has not been 
handled before.

For inner-join optimization, RDF engines like TriAD \cite{triad}, RDF-3X 
\cite{rdf3x}, gStore \cite{gstore} take the approaches like graph summarization, 
sideways-information-passing etc for an early pruning of triples. Systems like 
TripleBit \cite{triplebit} use a variable length bitwise encoding of RDF 
triples, and a query plan generation that favors queries with ``star'' joins, 
i.e., many triple patterns joining over a single variable. RDF engines built on 
top of commercial databases such as DB2RDF \cite{ibmsigmod13} propose creation 
of \textit{entity-oriented} flexible schemas and better data-flow techniques 
through the query plan to improve the performance of ``star'' join queries. 
Along with this, there are distributed RDF processing engines such as H-RDF-3X 
\cite{hrdf3x} and SHARD \cite{shard}.

While many of these engines mainly focus on efficient indexing of RDF graphs, 
BGP queries (inner-joins), and exploiting ``star'' patterns in the queries, we 
have focused on the broader components of SPARQL patterns such as OPTIONALs, 
UNIONs, FILTERs, and DISTINCT, which cannot always exploit the benefits of 
inner-join focused query optimizers.

\section{Conclusion}
In this article we have done an extensive analysis of a wide range of SPARQL 
components such as OPTIONAL patterns, UNIONs, FILTERs, DISTINCTs, and proposed 
that they can be evaluated by simply using our BGP-OPT pattern's optimization 
techniques as building blocks. We have extended the previously proposed 
concepts of \textit{minimality} of triples (tuples), \textit{cyclicity} of 
queries, and \textit{nullification}, \textit{best-match} operations. With this 
\textit{first of a kind} analysis of a wide range of SPARQL components, we hope 
to create novel optimization techniques for performance intensive SPARQL 
components. Our analysis shows that this is be possible by simple semantic 
manipulation of various intermixed query components. Our article also 
elaborately discusses the treatment of NULL values in the RDF and SPARQL 
context.

Since there is a very close resemblance between SPA-RQL components and SQL 
queries, our proposed techniques, as well as observations about query's 
structural properties and optimization opportunities, can be directly applicable 
to the respective SQL components too.

\end{document}